\documentclass[10pt,a4paper]{article}

\usepackage{amsmath,amssymb,amsthm,mathrsfs}
\usepackage{stmaryrd}  
\usepackage{graphicx}
\usepackage{xcolor}
\usepackage{soul}
\usepackage{fullpage}

\newlength{\defbaselineskip}
\newcommand{\setlinespacing}[1]%
           {\setlength{\baselineskip}{#1 \defbaselineskip}}

\theoremstyle{plain}
\newtheorem{theorem}{Theorem}[section]
\newtheorem{lemma}[theorem]{Lemma}

\newtheorem{corollary}[theorem]{Corollary}

\theoremstyle{definition}

\newtheorem{ass}[theorem]{Assumption}

\theoremstyle{remark}
\newtheorem{remark}[theorem]{Remark}

\numberwithin{equation}{section}

\DeclareMathOperator*{\esssup}{ess\,sup}

\newcommand{\fp}{\mathfrak{p}}
\newcommand{\fd}{\mathfrak{d}}

\newcommand{\bE}{\mathbb{E}}

\newcommand{\etamin}{\eta_\star}
\newcommand{\lambdamin}{\lambda_\star}

 \allowdisplaybreaks

\begin{document}

\title{Mean-Field Leader-Follower Games\\ with Terminal State Constraint\thanks{Financial support by the TRCRC 190 {\it Rationality and competition: the economic performance of individuals and firms} is gratefully acknowledged. We thank Paulwin Graewe for valuable discussions.}
}

\author{Guanxing Fu\footnote{Department of Mathematics, Humboldt-Universit\"at zu Berlin,
         Unter den Linden 6, 10099 Berlin, Germany; email: fuguanxi@math.hu-berlin.de}  ~ and ~ Ulrich Horst\footnote{Department of Mathematics, and School of Business and Economics, Humboldt-Universit\"at zu Berlin,
         Unter den Linden 6, 10099 Berlin, Germany; email: horst@math.hu-berlin.de}
}

\maketitle

\begin{abstract}
We analyze linear McKean-Vlasov forward-backward SDEs arising in leader-follower games with mean-field type control and terminal state constraints on the state process. We establish an existence and uniqueness of solutions result for such systems in time-weighted spaces as well as a {convergence} result of the solutions with respect to certain perturbations of the drivers of both the forward and the backward component. The general results are used to solve a novel  single-player model of portfolio liquidation under market impact with expectations feedback as well as a novel Stackelberg game of optimal portfolio liquidation with asymmetrically informed players.
\end{abstract}



\section{Introduction and overview}
Mean field games (MFGs) are a powerful tool to analyze strategic interactions in large populations when each individual player has only a small impact on the behavior of other players. Introduced independently by Huang, Malham\'e and Caines \cite{HMC-2006} and Lasry and Lions \cite{LL-2007}, MFGs have received considerable attention in the probability and stochastic control literature in the last decade. A probabilistic approach to solving MFGs was introduced by Carmona and Delarue in \cite{CD-2013}. Using a maximum principle of Pontryagin type, they showed that solving the MFG reduces to solving a McKean-Vlasov forward-backward SDE (FBSDE) of form,
\begin{equation}\label{FBSDE-probab-anal}
    \left\{\begin{split}
        dX_t=&~b(t,X_t,Y_t,\mathcal L(X_t,Y_t))\,dt+\sigma\,dW_t,\\
        -dY_t=&~h(t,X_t,Y_t,\mathcal L(X_t,Y_t))\,dt-Z_t\,dW_t,\\
        X_0=&~\chi,~Y_T=l(X_T,\mathcal L(X_T)),
    \end{split}\right.
\end{equation}
where $X$ is the state of the representative player, $Y$ is the adjoint variable, and $\mathcal L(\cdot)$ denotes the law of a stochastic process. In MFGs with common noise \cite{Ahuja-2016,ARY-2016} the dependence of the coefficients on the law of the process $(X,Y)$ is of conditional form. FBSDEs of the form \eqref{FBSDE-probab-anal} also arise in mean-field control (MFC) problems \cite{ABC-2018,AD-2011,CD-2015} and in MFGs with a major player \cite{BCY-2015,BCLY-2017,CZ-2016} when formulating stochastic maximum principles. MFGs with a major player are a special class of leader-follower games with mean-field control. In such a game, the leader's optimization problem can be viewed as MFC control problem where the state dynamics follows a controlled FBSDE that characterizes the representative minor agent's optimal response to the leader's control. We study a novel class of leader-follower games with mean-field control and terminal state constraint on the state processes that naturally arise in Stackelberg games of optimal portfolio liquidation with asymmetrically informed players.

\subsection{McKean-Vlasov FBSDE with terminal state constraint}
Let $W=(\overline W,W^0)$ be a {multi-dimensional} Brownian motion generating the filtration $\mathbb F=(\mathcal F_t)_{t\geq 0}$ and let $\mathbb F^0=(\mathcal F^0_t)_{t\geq 0}$ be the filtration generated by $W^0$. In this paper, we consider linear McKean-Vlasov FBSDEs of the form
\begin{equation}\label{FBSDE-general}
    \left\{\begin{aligned}
         dQ_t=&~\left(-\Lambda^1_tR_t-\Lambda^2_t\mathbb E\left[\left.\gamma_tQ_t\right|\mathcal F^0_t\right]+\overline f_t\right)\,dt,\\
        -dR_t=&~\left(\Lambda^4_t Q_t+\Lambda^3_t\mathbb E[\zeta_tR_t|\mathcal F^0_s]+\Lambda^5_t\mathbb E[\varrho_tQ_t|\mathcal F^0_t]+\overline g_t\right)\,dt-Z_t\,d W_t,\\
        Q_0=&~\chi, ~Q_T=0,
    \end{aligned}\right.
\end{equation}
with given initial and terminal condition for the forward, and unspecified terminal condition for the backward process. FBSDEs of this form arise in linear-quadratic MFGs, MFC problems, and leader-follower games under a terminal state constraint on the state process when formulating stochastic maximum principles. Under a terminal state constraint on the state sequence the terminal value of the adjoint process is unknown. The special case $\Lambda^2=\Lambda^3=\Lambda^5=\overline f=\overline g=0$ arises in the single player portfolio liquidation models under market impact studied in, e.g. \cite{AJK-2014,GHS-2013}. The special case $\Lambda^2=\Lambda^5=\overline f=\overline g=0$  was recently analyzed in \cite{FGHP-2018} in the framework of a MFG of optimal portfolio liquidation.

We prove a general existence and uniqueness of solutions result for the system \eqref{FBSDE-general} under boundedness assumptions on the model parameters that allows us to solve single player portfolio liquidation problems with private information and expectations feedback. The existence and uniqueness result is complemented by a convergence result for the solution of \eqref{FBSDE-general} with respect to the parameters $(\overline f,\overline  g)$ that allows us to formulate a stochastic maximum principle for leader-follower games of portfolio liquidation with asymmetrically informed players.

The existence and uniqueness of solutions to \eqref{FBSDE-general} is obtained via two nested continuation arguments. Standard continuation methods for McKean-Vlasov FBSDEs established in, e.g.~\cite{ARY-2016,BensoussanYamZhang15} do not apply to the system \eqref{FBSDE-general}, due to the unknown terminal value of the backward process. In order to overcome this problem we make a linear ansatz $R=AQ+H$, from which we derive an exogenous BSDE with singular terminal condition for the process $A$, and a BSDE with {known asymptotic behavior at the terminal time} for the process $H$. The driver of the latter BSDE depends on the unbounded process $A$. The nature of the FBSDE for $(Q,H)$ is different from \cite{FGHP-2018} where a similar ansatz gave a BSDE with known terminal condition.  Analyzing simultaneously the triple $(Q,H,R)$ allows us to prove the fixed-point condition arising in the application of the continuation method in a suitable space.

Our second main result is a convergence result for the solution $(Q,R)$ to the system \eqref{FBSDE-general} with respect to the ``input'' $(\overline f,\overline g)$. Our convergence is not in the $L^2$ sense as in the standard FBSDE literature \cite{MWZZ-2015,Yong-2010} but rather in the $L^\nu$ ($1<\nu<2$) sense. Specifically, we consider the convergence of the solutions $(Q^n,R^n)$ to a penalized version of \eqref{FBSDE-general} under a uniform $L^2$ boundedness assumption on the sequence $(\overline f^n,\overline g^n)$. For such inputs a result of Koml\'os \cite{Komlos-1967} guarantees the Cesaro convergence of $(\overline f^n,\overline g^n)$ along a subsequence in $L^\nu$ ($1<\nu<2$). We prove the convergence of the solutions in the same sense. To this end, we define auxiliary processes to decouple the system \eqref{FBSDE-general} and then show that these processes solve the system \eqref{FBSDE-general} in the right spaces. The convergence result then follows from the previously established uniqueness result.


\subsection{Applications to optimal portfolio liquidation}\label{section-motivation}
Models of optimal portfolio liquidation have received substantial attention in the financial mathematics and stochastic control literature in recent years; see \cite{AJK-2014,GH-2017,GHQ-2015,GHS-2013,KP-2016,PZ-2017} among many others. In such models, the controlled state sequence typically follows a dynamic of the form
\[
	X_t=x-\int_0^t\xi_s\,ds,
\]
where $x \in \mathbb{R}$ is the initial portfolio, and $\xi$ is the trading rate. The set of admissible controls is confined to those processes $\xi$ that satisfy almost surely the liquidation constraint $X_{T} = 0.$
It is typically assumed that the unaffected price process against which the trading costs are benchmarked follows some Brownian martingale $S$ and that the trader's transaction price is given by
\[
	\widetilde S_t=S_t-\int_0^t\kappa_s\xi_s\,ds-\eta_t\xi_t.
\]
The integral term accounts for permanent price impact; the term $\eta_t \xi_t$ accounts for instantaneous impact that does not affect future transactions. The trader's objective is then to minimize the cost functional
\[
	J(\xi) = \bE\left[\int_0^T\!\!\!\Big(\kappa_s \xi_s X_s + \eta_s|\xi_s|^2+\lambda_s|x_s|^2  \Big)\,ds \right]
\]
over all admissible liquidation strategies. We refer to \cite{AJK-2014, GHS-2013} for an interpretation of the processes $\eta, \kappa, \lambda$.


\subsubsection{Single player model with expectations feedback}\label{portfolio1}

Standard portfolio liquidation models assume that a trader's permanent price impact is driven by his observable transactions. If the transactions are not directly observable, then it is natural to assume that the permanent impact is driven by the market's expectation about the trader's transactions as in \cite{ABC-2018,BP-2017}, given the publicly observable information.

In Section 3 we solve a single-player liquidation model with expectations feedback where uncertainty is generated by the multi-dimensional Brownian motion $W=(\overline W,W^0)$. The Brownian motion $W^0$ describes a commonly observed random factor that drives market dynamics; the Brownian motion $\overline W$ is private information to the trader. Specifically, we assume that the trader's transaction price is given by
\begin{equation}\label{price-follower}
    \widetilde S_t=S_t-\int_0^t \left\{ \kappa_s\mathbb E[\xi_s|\mathcal F^0_s] + \widetilde g_s \right\}\,ds-\eta_t\xi_t,
\end{equation}
where $S$ is an $\mathbb F^0$ martingale, $\mathbb E[\xi_s|\mathcal F^0_s]$ is the market's expectation about the trader's strategy, and $\widetilde g$ is an $\mathbb F^0$-adapted process that will be endogenized in the next subsection. Assuming a standard quadratic running cost function as in \cite{AJK-2014, GHQ-2015,GHS-2013}, the objective of the trader is then to minimize the functional
\begin{equation}\label{cost-MKV}
    J(\xi)=\mathbb E\left[\int_0^T\kappa_tX_t\mathbb E[\xi_t|\mathcal F^0_t]+\widetilde g_t X_t + \eta_t\xi_t^2+\lambda_tX^2_t\,dt\right],
\end{equation}
 subject to the state dynamics
 \begin{equation}\label{state-MFC}
 \begin{split}
    dX_t &=-\xi_t\,dt\\
    X_0 &= x, ~~ X_T =0.
 \end{split}	
 \end{equation}
We allow the cost coefficients to be private information, i.e. to be $\mathbb F$ adapted. This justifies the conditional expectation term in the price dynamics. A standard stochastic maximum principle suggests that the optimal strategy is given by
 \begin{equation}
 	\xi^*_t = \frac{Y_t-\mathbb E[\kappa_tX_t|\mathcal F^0_t]}{2\eta_t},
 \end{equation}
 where $X$ is the portfolio process, $Y$ is the adjoint variable, and $(X,Y)$ solves  \eqref{FBSDE-general} with $\overline f=0,~ \overline g = \widetilde g$:
\begin{equation}\label{MKV-FBSDE}
\left\{\begin{aligned}
dX_t =&~-\frac{Y_t-\mathbb E[\kappa_tX_t|\mathcal F^0_t]}{2\eta_t}\,dt,\\
-dY_t =&~\left(\kappa_t\mathbb E\left[\left.\frac{Y_t}{2\eta_t}\right|\mathcal F^0_t\right]-\kappa_t\mathbb E\left[\left.\frac{1}{2\eta_t}\right|\mathcal F^0_t\right]\mathbb E[\kappa_tX_t|\mathcal F^0_t]+2\lambda_tX_t+\widetilde g_t\right)\,dt-Z_t\,d W_t,\\
X_0 =&~x,~X_T = ~0.
\end{aligned}
\right.
\end{equation}

If the terms $\mathbb E[\kappa_tX_t|\mathcal F^0_t]$ and $\kappa_t\mathbb E\left[\left.\frac{1}{2\eta_t}\right|\mathcal F^0_t\right]\mathbb E[\kappa_tX_t|\mathcal F^0_t]$ drop out of the FBSDE system, then the system reduces to that arising in the MFG analyzed in \cite{FGHP-2018}. In the next subsection we introduce a model extension where the privately informed trader is the follower in a Stackelberg game of optimal portfolio liquidation. As a byproduct we obtain an extension of the MFG in \cite{FGHP-2018} to a MFG with a major player. A related model without liquidation constraint and without any feedback of the major player's strategy on the minor players' transaction price has been considered in \cite{HJN-2015}.


\subsubsection{Mean-Field type Stackelberg game with asymmetric information}\label{portfolio2}

In Section 4 we solve a Stackelberg game of optimal portfolio liquidation with asymmetrically informed players. The leader (she) has the first-mover advantage while the follower (he) has an informational advantage.

We assume again that uncertainty is generated by the multi-dimensional Brownian motion $W=(\overline W,W^0)$ and that $W^0$ describes a commonly observed market factor while $\overline W$ is private information to the follower. For a given ${\mathbb F}^0$-adapted strategy $\xi^0$ of the Stackelberg leader, we assume that the follower's liquidation problem is the same as in the previous subsection with $\widetilde g=\widetilde\kappa^0\xi^0$ for some $\mathbb F^0$-adapted process $\widetilde\kappa^0$. Let $\xi^*(\cdot)$ be the follower's optimal response function to the leader's strategy and put $\mu^* := \mathbb E[\xi^*(\cdot)|{\cal F}^0]$. Following the standard approach we assume that the leader's transaction price is
\begin{equation}\label{trading-price-leader}
        \widetilde S^0_t=S_t-\int_0^t\overline\kappa^0_s\mu^*_s\,ds-\int_0^t\kappa^0_s\xi^0_s\,ds-\eta^0_t\xi^0_t
\end{equation}
for ${\mathbb F}^0$-adapted coefficients $\eta^0,\kappa^0, \overline \kappa^0$. The difference is that now the leader controls the transaction price both directly and indirectly through the dependence of the follower's optimal response on her trading strategy. We furthermore assume that the leader's cost functional is given by
\begin{equation} \label{leader-cost}
	J^0(\xi^0)=\mathbb E\left[\int_0^T\left(\overline\kappa^0_t\mu^*_tX^0_t+\kappa^0_tX^0_t\xi^0_t+\eta^0_t(\xi^0_t)^2+\lambda^0_t (X^0_t)^2+\overline\lambda_t(\mu_t^*)^2\right)\,dt\right],
\end{equation}
where $X^0$ denotes her portfolio process and $\lambda^0, \overline \lambda$ are ${\mathbb F}^0$-adapted. Her control problem is then a MFC problem with state process $(X^0,X,Y)$, where $(X,Y)$ solves \eqref{MKV-FBSDE} with $\widetilde g=\widetilde\kappa^0\xi^0$ and
\begin{equation}\label{state-leader}
\begin{split}
    dX^ 0_t &=-\xi^ 0_t\,dt\\
    X^ 0_0 &= x, ~ X^ 0_T =0.
 \end{split}	
\end{equation}
We establish a new maximum principle for this control problem from which we derive an explicit representation of the major player's optimal control $\xi^{0,*}$ as
\begin{equation} \label{leader-optimal}
	 \xi^{0,*}_t=\frac{p_t+\mathbb E[\widetilde\kappa^0_tq_t|\mathcal F^0_t]-\kappa^0_tX^{0,*}_t}{2\eta^0_t}
\end{equation}
in terms of the state equation \eqref{state-leader} and the adjoint equations:
\begin{equation}\label{leader-adjoint-p}
    -dp_t=\left(\overline\kappa^0_t\mathbb E\left[\left.\frac{Y^{}_t}{2\eta_t}\right|\mathcal F^0_t\right]-\overline\kappa^0_t\mathbb E\left[\left.\frac{1}{2\eta_t}\right|\mathcal F^0_t\right]\mathbb E[\kappa_tX_t|\mathcal F^0_t]+\kappa^0_t\xi^{0}_t+2\lambda^0_tX^{0}_t\right)\,dt-Z_t\,dW^0_t
\end{equation}
and
\begin{equation}\label{leader-adjoint-q-r}
    \left\{\begin{aligned}
        -dq_t=&~\left(-\frac{r_t}{2\eta_t}-\mathbb E\left[\kappa_tq_t|\mathcal F^0_t\right]\frac{1}{2\eta_t}+\overline f_t\right)\,dt,\\
                -dr_t=&~\left(-2\lambda_tq_t+\kappa_t\mathbb E\left[\left.\frac{r_t}{2\eta_t}\right|\mathcal F^0_t\right]+\kappa_t\mathbb E\left[\left.\frac{1}{2\eta_t}\right|\mathcal F^0_t\right]\mathbb E[\kappa_tq_t|\mathcal F^0_t]+\overline g_t\right)\,dt-Z_t\,dW_t,\\
        q_0=&~0,~q_T=0,
    \end{aligned}\right.
\end{equation}
where \[
    \overline f_t=\frac{\overline\kappa^0_tX^{0}_t}{2\eta_t}+\frac{\overline\lambda_t}{\eta_t}\mathbb E\left[\left.\frac{Y^{}_t}{2\eta_t}\right|\mathcal F^0_t\right]-\frac{\overline\lambda_t}{\eta_t}\mathbb E\left[\left.\frac{1}{2\eta_t}\right|\mathcal F^0_t\right]\mathbb E[\kappa_tX_t|\mathcal F^0_t]
\]
and
\[
    \overline g_t=-\kappa_t\mathbb E\left[\left.\frac{1}{2\eta_t}\right|\mathcal F^0_t\right]\overline\kappa^0_tX^0_t-2\overline\lambda_t\kappa_t\mathbb E\left[\left.\frac{1}{2\eta_t}\right|\mathcal F^0_t\right]\left(\mathbb E\left[\left.\frac{Y_t}{2\eta_t}\right|\mathcal F^0_t\right]-\mathbb E\left[\left.\frac{1}{2\eta_t}\right|\mathcal F^0_t\right]\mathbb E[\kappa_tX_t|\mathcal F^0_t]\right).
\]
Here, $p$ is the adjoint variable to $X^0$ and $(q,r)$ are the adjoint variables to $(Y,X)$.
The system \eqref{leader-adjoint-q-r} is again a special case of \eqref{FBSDE-general}.

In order to establish our maximum principle we first consider a sequence of unconstrained optimization problems where the liquidation constraints are replaced by increasingly penalized open positions at the terminal time. The resulting optimal strategies for the Stackelberg leader turn out to be $L^2$ bounded, hence they have Cesaro convergent subsequence. From this we deduce that the sequence of state-adjoint equations for the penalized problems Cesaro converges to the system \eqref{MKV-FBSDE}, \eqref{state-leader}, \eqref{leader-adjoint-p} and \eqref{leader-adjoint-q-r}.

The rest of this paper is organized as follows. Our general existence, uniqueness and convergence results for the FBSDE \eqref{FBSDE-general} are established in Section \ref{section-ex-sta}. The MFC problem and the Stackelberg game of optimal portfolio liquidation introduced above are solved in Section \ref{section-MKV} and Section \ref{section-MFG}, respectively.

{\sc Notation and conventions.} Throughout, we work on probability space $(\Omega,\mathbb P,\mathcal F)$, on which there exist two independent Brownian motions $W^0$ and $\overline W$. We denote by $\mathbb F^0=(\mathcal F^0_t)_{0\leq t\leq T}$ and $\mathbb F=(\mathcal F_t)_{0\leq t\leq T}$ the filtrations generated by $W^0$ and $W$, augmented by the $\mathbb P$ null sets, respectively, where $W=(\overline W,W^0)$.
For a space $\mathbb I$ and a filtration $\mathbb G$ we introduce the following spaces:
\begin{align*}
    L^0_{\mathbb G}([0,T]\times\Omega;\mathbb I)=& \{X: X:[0,T]\times\Omega\rightarrow\mathbb I\textrm{ and }X\textrm{ is }\mathbb G\textrm{ progressively measurable and }\mathbb I\textrm{ valued}\} \\
    L^k_{\mathbb G}([0,T]\times\Omega;\mathbb I)=&\left\{X\in L^0_\mathbb G([0,T]\times\Omega;\mathbb I):\mathbb E\left[\int_0^T|X_t|^k\,dt\right]<\infty\right\},~k\geq 1\\
    L^\infty_{\mathbb G}([0,T]\times\Omega;\mathbb I)=&\left\{X\in L^0_\mathbb G([0,T]\times\Omega;\mathbb I):\esssup_{(t,\omega)\in[0,T]\times\Omega}|X_t(\omega)|<\infty\right\}.
\end{align*}
The spaces $L^k_{\mathbb G}$ are equipped the norm $\|X\|_{L^k}=\left(\mathbb E\left[\int_0^T|X_t|^k\,dt\right]\right)^{1/k}$. The spaces
\begin{equation*}
\begin{split}
     S^2_{\mathbb G}([0,T]\times\Omega;\mathbb I)=& \left\{X\in L^0_{\mathbb G}([0,T]\times\Omega;\mathbb I): 
     \mathbb E\left[\sup_{0\leq t\leq T}|X_t|^2\right]<\infty\right\} \\
      S^{2,-}_{\mathbb G}([0,T)\times\Omega;\mathbb I)=& \left\{X\in L^0_{\mathbb G}([0,T)\times\Omega;\mathbb I):\sup_{\epsilon > 0}\mathbb E\left[\sup_{0\leq t\leq T-\epsilon}|X_t|^2\right] \leq C\right\}
\end{split}
\end{equation*}
are equipped with the respective norms
\[
	\|X\|_{S^2}:=\left(\mathbb E\left[\sup_{0\leq t\leq T}|X_t|^2\right]\right)^{1/2};\|X\|_{S^{2,-}}:=\sup_{\epsilon\geq 0}\left(\mathbb E\left[\sup_{0\leq t\leq T-\epsilon}|X_t|^2\right]\right)^{1/2},
\]	
and for $\beta>0$ we introduce the space
\[
   {\mathcal H}_{\beta}=\left\{X\in \mathbb S^2_{\mathbb F}([0,T]\times\Omega;\mathbb I): E\left[\sup_{t\in[0,T]}\left|\frac{|X_t|}{(T-t)^{\beta}}\right|^2\right]<\infty\right\}  \mbox{ with }
   \|X\|_\beta:=\left(\mathbb E\left[\sup_{0\leq t\leq T}\left|\frac{X_t}{(T-t)^\beta}\right|^2\right]\right)^{1/2}.
\]
For $\phi\in L^\infty_{\mathbb G}([0,T]\times\Omega;\mathbb I)$, we denote by $\|\phi\|$ and $\phi_\star$ its upper and lower bounds, respectively. Finally, we adopt the convention that a positive constant $C$ may vary from line to line.


\section{The McKean-Vlasov FBSDE}\label{section-ex-sta}
In this section, we prove a general existence and uniqueness of solutions result (in a suitable space) for the FBSDE \eqref{FBSDE-general} along with the convergence result with respect to the processes $(\overline  f,\overline  g)$. We assume throughout that the system coefficients satisfy the following assumption.

\begin{ass}\label{ass-ex-sta}
\begin{itemize}
    \item[i)] The stochastic processes $\gamma$, $\zeta$, {$\varrho$} and $\Lambda^i$ $(i=1,\cdots,5)$ belong to $L^\infty_{\mathbb F}$.
    \item[ii)] There exist constants $\theta_i>0$ $(i=1,2)$ such that
    \[
        \left(\Lambda^1-\frac{\|\gamma\||\Lambda^2|^2}{2\theta_1}-\frac{\|\Lambda^3\||\zeta|^2}{2\theta_2}\right)_\star>0
    \]
    and
    \[
        \left(\Lambda^4-\frac{\|\gamma\|\theta_1}{2}-\frac{\|\Lambda^3\|\theta_2}{2}-{\|\Lambda^5\|\|\varrho\|}\right)_\star>0.
    \]
    \item[iii)] The initial condition $\chi$ belongs to $L^2_{\mathbb F}$ and $(\overline f,\overline g)\in S^2_{\mathbb F}\times L^2_{\mathbb F}$.
\end{itemize}
\end{ass}

The linear ansatz $R= AQ+H$ on $[0,T)$ results in the following FBSDE for the triple $(Q,H,R)$:
{
\begin{equation}\label{FBSDE-appendix}
    \left\{\begin{aligned}
         dQ_t=&~\left(-\Lambda^1_tR_t-\Lambda^2_t\mathbb E\left[\left.\gamma_tQ_t\right|\mathcal F^0_t\right]+\overline f_t\right)\,dt,\\
        -dH_t=&~\left(-\Lambda^1_tA_tH_t-\Lambda^2_tA_t\mathbb E[\gamma_tQ_t|\mathcal F^0_t]+A_t\overline f_t+\Lambda^3_t\mathbb E[\zeta_tR_t|\mathcal F^0_t] \right. \\ & \left. \qquad +\Lambda^5_t\mathbb E[\varrho_tQ_t|\mathcal F^0_t]+\overline g_t\right)\,dt  -Z_t\,d W_t,\\
        -dR_t=&~\left(\Lambda^4_t Q_t+\Lambda^3_t\mathbb E[\zeta_tR_t|\mathcal F^0_s]+\Lambda^5_t\mathbb E[\varrho_tQ_t|\mathcal F^0_t]+\overline g_t\right)\,dt-Z_t\,d W_t,\\
        R=&~AQ+H,~t\in[0,T),\\
        Q_0=&~\chi,~Q_T=0,
    \end{aligned}\right.
\end{equation}}
where $A$ satisfies the singular BSDE
\begin{equation}\label{singular-BSDE}
    -dA_t=\left(\Lambda^4_t-\Lambda^1_tA^2_t\right)\,dt-Z_t\,d W_t,\quad \lim_{t\nearrow T}A_t=\infty.
\end{equation}
It has been shown in \cite{AJK-2014,GHS-2013} that the equation \eqref{singular-BSDE} is well-posed under Assumption \ref{ass-ex-sta} and that the following estimate holds:
\begin{equation}\label{estimate-A}
     \frac{1}{\mathbb E\left[\left.\int_t^T\Lambda^1_u\,du\right|\mathcal F_t\right]}\leq A_t\leq \frac{1}{(T-t)^2}\mathbb E\left[\left.\int_t^T\frac{1}{\Lambda^1_u}+(T-u)^2\Lambda^4_u\,du\right|\mathcal F_t\right].
\end{equation}
It follows from \eqref{estimate-A} that $A$ is nonnegative and that for all $0\leq t_1<t_2\leq T$,
\begin{equation}\label{estimate--exp-bar-C-appendix}
    e^{-\int_{t_1}^{t_2}\Lambda^1_sA_s\,ds}\leq C\left(\frac{T-t_2}{T-t_1}\right)^\beta\leq C\left(\frac{T-t_2}{T-t_1}\right)^\tau, \quad \mbox{where} \quad \beta:=\Lambda^1_\star/\|\Lambda^1\| \mbox{ and } 0\leq\tau\leq\beta.
\end{equation}


\subsection{Existence and uniqueness of solutions}

In view of \cite{FGHP-2018}, we expect to find a solution $(Q,H,R)$ to \eqref{FBSDE-appendix} such that $(Q,R) \in \mathcal H_\alpha\times L^2_{\mathbb F}$ for some $\alpha>0$. Unlike in \cite{FGHP-2018} the process $H$ is only defined on $[0,T)$. The following heuristics suggests that if we can find a solution such that $(Q,R) \in \mathcal H_\alpha\times L^2_{\mathbb F}$, then $H \in S^{2,-}_{\mathbb F}$. In fact, by the general solution formula for linear BSDEs,
for any $0\leq t<\widetilde T<T$,
\[
    H_t=\mathbb E\left[\left.H_{\widetilde T}e^{-\int_t^{\widetilde T}\Lambda^1_uA_u\,du}+\int_t^{\widetilde T}e^{-\int_t^s\Lambda^1_uA_u\,du}K_s\,ds\right|\mathcal F_t\right],
\]
where
\[
    K_s=\left(-\Lambda^2_sA_s\mathbb E[\gamma_sQ_s|\mathcal F^0_s]+A_s\overline f_s+\Lambda^3_s\mathbb E[\zeta_sR_s|\mathcal F^0_s]+{\Lambda^5_s\mathbb E[\varrho_sQ_s|\mathcal F^0_s]}+\overline g_s\right).
\]
If we knew that
\begin{equation}\label{upper-limit-H}
    \limsup_{\widetilde T\nearrow T}\mathbb E[|H_{\widetilde T}|^2]<\infty,
\end{equation}
then taking the limit $\widetilde T\nearrow T$ and using the estimate \eqref{estimate--exp-bar-C-appendix},
\begin{equation}\label{integral-H}
    H_t=\mathbb E\left[\left.\int_t^{T}e^{-\int_t^s\Lambda^1_uA_u\,du}K_s\,ds\right|\mathcal F_t\right].
\end{equation}
From this and using \eqref{estimate--exp-bar-C-appendix} again, we obtain a constant $C>0$ such that for any $\epsilon>0$,
\begin{equation*}
    \mathbb E\left[\sup_{0\leq t\leq T-\epsilon}|H_t|^2\right]\leq C\left(\|Q\|_\alpha+\|\overline f\|_{S^2}+\|R\|_{L^2}+\|\overline g\|_{L^2}\right).
\end{equation*}
Since \eqref{upper-limit-H} holds for $H \in S^{2,-}_{\mathbb F}$ our goal is to establish the existence and  uniqueness of a solution $(Q,H,R) \in \mathcal H_\alpha\times S^{2,-}_{\mathbb F}\times L^2_{\mathbb F}$. To this end, we apply a nested continuation method to the system:
\begin{equation}\label{continuation-r-q}
    \left\{\begin{aligned}
         dQ_t=&~\left(-\Lambda^1_tR_t-\Lambda^2_t\mathbb E\left[\left.\gamma_tQ_t\right|\mathcal F^0_t\right]+\overline f_t\right)\,dt,\\
        -dH_t=&~\left(-\Lambda^1_t A_tH_t-\Lambda^2_tA_t\mathbb E[\gamma_tQ_t|\mathcal F^0_t]+A_t\overline f_t+\fp\Lambda^3_t\mathbb E[\zeta_tR_t|\mathcal F^0_t] \right. \\ & \left. \qquad \qquad +\fp\Lambda^5_t\mathbb E[\varrho_tQ_t|\mathcal F^0_t]+\overline g_t+f_t\right)\,dt-Z_t\,dW_t,\\
        -dR_t=&~\left(\Lambda^4_t Q_t+\fp\Lambda^3_t\mathbb E[\zeta_tR_t|\mathcal F^0_s]+\fp\Lambda^5_t\mathbb E[\varrho_tQ_t|\mathcal F^0_t]+\overline g_t+f_t\right)\,dt-Z_t\,d W_t,\\
        R=&~ AQ+H,~t\in[0,T),\\
        Q_0=&~\chi,~Q_T=0.
    \end{aligned}\right.
    \end{equation}

In a first step, we prove the existence of a unique solution to the above system for $\fp=0$. Subsequently, we show that the solution result extends to $\fp=1$.

\begin{lemma}\label{appendix-p-0}
If  $\fp=0$, then the FBSDE \eqref{continuation-r-q} has a solution in $\mathcal H_\alpha\times S^{2,-}_{\mathbb F}\times L^2_{\mathbb F}$ for any $f\in L^2_{\mathbb F}$, where $0<\alpha<\beta$.
\end{lemma}
\begin{proof}
Notice that the system \eqref{continuation-r-q} is still coupled for $\fp=0$. To solve it, we apply a continuation method to the following system:
\begin{equation}\label{appendix-p-0-p-bar}
    \left\{\begin{split}
         dQ_t=&~\left(-\Lambda^1_tR_t-\overline\fp\Lambda^2_t\mathbb E\left[\left.\gamma_tQ_t\right|\mathcal F^0_t\right]+\overline f_t+b'_t\right)\,dt,\\
        -dH_t=&~\left(-\Lambda^1_t A_tH_t-\overline\fp\Lambda^2_tA_t\mathbb E[\gamma_tQ_t|\mathcal F^0_t]+A_t\overline f_t+\overline g_t+f_t+f'_t\right)\,dt-Z_t\,dW_t,\\
        -dR_t=&~\left(\Lambda^4_t Q_t+\overline g_t+f_t+f'_t-A_tb'_t\right)\,dt-Z_t\,d W_t,\\
        R=&~  AQ+H,~t\in[0,T),\\
        Q_0=&~\chi, ~Q_T=0.
    \end{split}\right.
\end{equation}

\textbf{Step 1.} {\it For $\overline\fp=0$, the system \eqref{appendix-p-0-p-bar} is solvable in $\mathcal H_\alpha \times S^{2,-}_{\mathbb F} \times L^2_{\mathbb F}$ for any $(b',f')\in\mathcal H_{\alpha}\times\mathcal H_{\alpha-1}$.}

If $\overline\fp=0$, then the system \eqref{appendix-p-0-p-bar} is decoupled and we let $H$ be
\begin{equation}\label{remark-equation-1}
    \begin{split}
    H_t=&~\mathbb E\left[\left.\int_t^Te^{-\int_t^s\Lambda^1_u A_u\,du}\left(A_s\overline f_s+\overline g_s+f_s+f'_s\right)\,ds\right|\mathcal F_t\right],~0\leq t<T.
    \end{split}
\end{equation}
Moreover, by the estimate \eqref{estimate--exp-bar-C-appendix} and Doob's maximal inequality, we have for any $\epsilon>0$,
\begin{equation}\label{p-bar-0-H-estimate}
    \mathbb E\left[\sup_{0\leq t\leq T-\epsilon}|H_t|^2\right]\leq C\left(\|\overline f\|_{S^2}+\|\overline g\|_{L^2}+\|f\|_{L^2}+\|f'\|_{\alpha-1}\right),
\end{equation}
where $C$ is independent of $\epsilon$. Thus, $H$ belongs to $S^{2,-}_{\mathbb F}$ and satisfies the SDE in \eqref{appendix-p-0-p-bar}.

We now turn to the process $Q$.
Taking $R=AQ+H$ into the SDE for $Q$ yields,
\begin{equation}\label{remark-equation-2}
    Q_t=\chi e^{-\int_0^t\Lambda^1_uA_u\,du}+\int_0^te^{-\int_s^t\Lambda^1_uA_u\,du}\left(-\Lambda^1_sH_s+\overline f_s+b'_s\right)\,ds,~0\leq t\leq T.
\end{equation}
Using monotone convergence and the estimate \eqref{p-bar-0-H-estimate} this implies,
\begin{equation}\label{p-bar-0-Q-S2-T-epsilon}
    \begin{split}
    &~\mathbb E\left[\sup_{0\leq t\leq T}\left|\frac{Q_t}{(T-t)^\alpha}\right|^2\right]\\
    \leq&~ C\left(\|\chi\|_{L^2}+\mathbb E\left[\left(\int_0^{T}\frac{|H_s|}{(T-s)^\alpha}\,ds\right)^2\right]+\|\overline f\|_{S^2}+\|b'\|_{\alpha}\right)\\
    =&~ C\left(\|\chi\|_{L^2}+\lim_{\epsilon\searrow 0}\mathbb E\left[\left(\int_0^{T-\epsilon}\frac{|H_s|}{(T-s)^\alpha}\,ds\right)^2\right]+\|\overline f\|_{S^2}+\|b'\|_{\alpha}\right) \\
    \leq&~C\left(\|\chi\|_{L^2}+\lim_{\epsilon\searrow 0}\mathbb E\left[\sup_{0\leq t\leq T-\epsilon}|H_t|^2\right]+\|\overline f\|_{S^2}+\|b'\|_{\alpha}\right)\\
    \leq&~C\left(\|\chi\|_{L^2}+\|\overline f\|_{S^2}+\|\overline g\|_{L^2}+\|f\|_{L^2}+\|f'\|_{\alpha-1}+\|b'\|_{\alpha}\right).
    \end{split}
\end{equation}
This shows that $Q\in\mathcal H_\alpha$.
Integration by parts for the product $QR$ on $[0,T-\epsilon]$ yields,
\begin{equation*}
    \begin{split}
        H_{T-\epsilon}Q_{T-\epsilon} & \leq A_{T-\epsilon}Q_{T-\epsilon}^2+H_{T-\epsilon}Q_{T-\epsilon}
        =Q_{T-\epsilon}R_{T-\epsilon}\\
        & \leq -\int_0^{T-\epsilon}\left(Q^2_t+R^2_t\right)\,dt+CA_0\chi^2+|\chi H_0|+C\int_0^{T-\epsilon}|Q_t||\overline g_t+f_t+f'_t+A_tb'_t|\,dt \\ & \qquad +\int_0^{T-\epsilon}Q_tZ_t\,d\overline W_t.
    \end{split}
\end{equation*}
Taking expectations on both sides we have
\begin{equation*}
    \begin{split}
        &~\mathbb E\left[\int_0^{T-\epsilon}\left(Q^2_t+R^2_t\right)\,dt\right]\\
        \leq&~ \mathbb E[CA_0\chi^2]+\mathbb E[|\chi H_0|]+C\mathbb E\left[\int_0^{T-\epsilon}|Q_t||\overline g_t+f_t+f'_t+A_tb'_t|\,dt\right]+\mathbb E[|H_{T-\epsilon}Q_{T-\epsilon}|]\\
        \leq&~C\left(\mathbb E[A_0\chi^2]+C\mathbb E[|\chi H_0|]+\|Q\|_\alpha\right)\\
        &~+C\left(\|\overline g\|_{L^2}+\|f\|_{L^2}+\|f'\|_{\alpha-1}+\|b'\|_\alpha\right)+\mathbb E[|H_{T-\epsilon}Q_{T-\epsilon}|],
    \end{split}
\end{equation*}
Thus, by taking $\epsilon\rightarrow 0$, from \eqref{estimate-A}, \eqref{p-bar-0-H-estimate} and \eqref{p-bar-0-Q-S2-T-epsilon} we get $R\in L^2_{\mathbb F}$.

\textbf{Step 2.} {\it If \eqref{appendix-p-0-p-bar} admits a solution for some $\overline\fp\in[0,1]$ and for any $(b',f')\in\mathcal H_{\alpha}\times\mathcal H_{\alpha-1}$, then the same holds for $\overline\fp+\overline\fd$ for some constant $\overline\fd$} that does not depend on $\overline\fp$.

For fixed $Q\in\mathcal H_\alpha$, since
\[
    -\overline\fd\Lambda^2\mathbb E\left[\left.\gamma Q \right|\mathcal F^0\right]+b'\in\mathcal H_\alpha,\quad -\overline\fd\Lambda^2A\mathbb E[\gamma Q|\mathcal F^0]+f'\in\mathcal H_{\alpha-1},
\]
there exists a solution $(\widetilde Q,\widetilde H,\widetilde R)\in\mathcal H_\alpha\times S^{2,-}_{\mathbb F}\times L^2_{\mathbb F}$ to the following system:
\begin{equation}\label{appendix-p-0-p-bar-d-bar}
    \left\{\begin{split}
         d\widetilde Q_t=&~\left(-\Lambda^1_t\widetilde R_t-\overline\fp\Lambda^2_t\mathbb E\left[\left.\gamma_t\widetilde Q_t\right|\mathcal F^0_t\right]-\overline\fd\Lambda^2_t\mathbb E\left[\left.\gamma_t Q_t\right|\mathcal F^0_t\right]+\overline f_t+b'_t\right)\,dt,\\
       -d\widetilde H_t=&~\left(-\Lambda^1_t A_t\widetilde H_t-\overline\fp\Lambda^2_tA_t\mathbb E[\gamma_t\widetilde Q_t|\mathcal F^0_t] - \overline\fd\Lambda^2_tA_t\mathbb E[\gamma_t Q_t|\mathcal F^0_t] \right. \\ & \left. \qquad \qquad  +A_t\overline f_t+\overline g_t+f_t+f'_t\right)\,dt-Z_t\,d W_t,\\
        -d\widetilde R_t=&~\left(\Lambda^4_t \widetilde Q_t+\overline g_t+f_t+f'_t-A_tb'_t\right)\,dt-Z_t\,d W_t,\\
        \widetilde R=&~ A\widetilde Q+\widetilde H,~t\in[0,T),\\
        \widetilde Q_0=&~\chi, ~\widetilde Q_T=0.
    \end{split}\right.
\end{equation}
It remains to prove that the mapping $\Phi: \mathcal H_\alpha \rightarrow \mathcal H_\alpha$, $Q \mapsto \widetilde Q$ is a contraction when $\overline\fd$ is small enough and independent of $\overline\fp$.
For any $Q,~Q'\in\mathcal H_\alpha$, let $(\widetilde Q,\widetilde H,\widetilde R)$ and $(\widetilde Q',\widetilde H',\widetilde R')$ be the corresponding solutions. Integration by parts for $(\widetilde Q-\widetilde Q')(\widetilde R-\widetilde R')$ on $[0,T-\epsilon]$ implies,
\begin{equation*}
    \begin{split}
       &~ \mathbb E\left[\int_0^{T-\epsilon}\left(\Lambda^4_s-\frac{\|\gamma\|\theta_1}{2}\right)(\widetilde Q_s-\widetilde Q'_s)^2\,dts\right]+\mathbb E\left[\int_0^{T-\epsilon}\left(\Lambda^1_s-\frac{\|\gamma\|(\Lambda^2_s)^2}{2\theta^1}\right)(\widetilde R_s-\widetilde R'_s)^2\,dt\right]\\
       \leq&~C\mathbb E\left[|\widetilde Q_{T-\epsilon}\widetilde H_{T-\epsilon}|\right]+\varepsilon\mathbb E\left[\int_0^{T-\epsilon}\left(\widetilde R_s-\widetilde R'_s\right)^2\,ds\right]+C\overline\fd\mathbb E\left[\int_0^{T-\epsilon}(Q_t-Q'_t)^2\,dt\right].
    \end{split}
\end{equation*}
Letting $\epsilon\rightarrow 0$ and choosing $\varepsilon$ small enough, Assumption \ref{ass-ex-sta} yields,
\begin{equation}\label{p-bar-0-L2-Q}
    \mathbb E\left[\int_0^{T}(\widetilde Q_s-\widetilde Q'_s)^2\,dts\right]+\mathbb E\left[\int_0^{T}(\widetilde R_s-\widetilde R'_s)^2\,ds\right]\leq C\overline\fd\mathbb E\left[\int_0^{T}(Q_t-Q'_t)^2\,dt\right].
\end{equation}
Considering the SDE for $\widetilde Q$ in terms of $\widetilde R$, by \eqref{p-bar-0-L2-Q} we have
\begin{equation}\label{p-bar-0-S2-Q}
    \mathbb E\left[\sup_{0\leq t\leq T}|\widetilde Q_t-\widetilde Q'_t|^2\right]\leq C\overline\fd\mathbb E\left[\int_0^{T}(Q_t-Q'_t)^2\,dt\right].
\end{equation}
Since $\widetilde H\in S^{2,-}_{\mathbb F}$, we have the following expression:
\begin{equation}\label{p-bar-widetildetilde-H}
    \begin{split}
    \widetilde H_t=&~\mathbb E\left[\left.\int_t^Te^{-\int_t^s\Lambda^1_uA_u\,du}\left(-\overline\fp\Lambda^2_tA_t\mathbb E[\gamma_t\widetilde Q_t|\mathcal F^0_t] - \overline\fd\Lambda^2_tA_t\mathbb E[\gamma_t Q_t|\mathcal F^0_t] \right.\right. \right.\\
    &~\left.\left. \left.\left. \qquad \qquad  +A_t\overline f_t+\overline g_t+f_t+f'_t\right)\,dt\right)\right|\mathcal F_t\right].
    \end{split}
\end{equation}
From \eqref{p-bar-widetildetilde-H}, Doob's maximal inequality and \eqref{p-bar-0-S2-Q} yield that for any $\epsilon>0$
\begin{equation}\label{p-bar-0-S2-H}
    \begin{split}
        & ~ \mathbb E\left[\sup_{0\leq t\leq T-\epsilon}|\widetilde H_t-\widetilde H'_t|^2\right] \\
        \leq&~ C\mathbb E\left\{\sup_{0\leq t\leq T}\left|\mathbb E\left[\int_t^T\frac{(T-s)^{\beta-1}}{(T-t)^\beta}\mathbb E[\left.|\widetilde Q_s-\widetilde Q'_s||\mathcal F^0_s]\,ds\right|\mathcal F_t\right]\right|^2\right\}\\
        &~+C\overline\fd\mathbb E\left\{\sup_{0\leq t\leq T}\left|\mathbb E\left[\int_t^T\frac{(T-s)^{\beta-1}}{(T-t)^\beta}\mathbb E[\left.|Q_s-Q'_s||\mathcal F^0_s]\,ds\right|\mathcal F_t\right]\right|^2\right\}\\
        \leq&~C\mathbb E\left\{\sup_{0\leq t\leq T}\left|\mathbb E\left[\sup_{0\leq s\leq T}\mathbb E[\left.|\widetilde Q_s-\widetilde Q'_s||\mathcal F^0_s]\right|\mathcal F_t\right]\right|^2\right\}\\
        &~+C\mathbb E\left\{\sup_{0\leq t\leq T}\left|\mathbb E\left[\sup_{0\leq s\leq T}\mathbb E[\left.|Q_s-Q'_s||\mathcal F^0_s]\right|\mathcal F_t\right]\right|^2\right\}\\
        \leq&~C\overline\fd\mathbb E\left[\int_0^T(Q_t-Q'_t)^2\,dt\right]+C\overline\fd\mathbb E\left[\sup_{0\leq t\leq T}\left|Q_t-Q'_t\right|^2\right],
    \end{split}
\end{equation}
where $C$ is independent of $\epsilon$.
Finally, considering the SDE for $\widetilde Q$ in terms of $\widetilde H$, by \eqref{p-bar-0-S2-Q}, \eqref{p-bar-0-S2-H} and the same argument as \eqref{p-bar-0-Q-S2-T-epsilon}, we have
\[
    \|\widetilde Q-\widetilde Q'\|_\alpha\leq C\overline\fd\|Q-Q'\|_\alpha.
\]
Thus, when $\overline\fd$ is small enough, $\Phi$ is a contraction. Iterating the argument finitely often and letting $f'=b'=0$ yields the desired result.
\end{proof}

\begin{theorem}\label{existence-appendix}
The FBSDE system \eqref{FBSDE-appendix} admits a unique solution $(Q,H,R)\in\mathcal H_\alpha\times S^{2,-}_{\mathbb F}\times L^2_{\mathbb F}$, where $0<\alpha<\beta$.
\end{theorem}
\begin{proof}
We first prove the existence of a solution. In a second step we prove the uniqueness of solutions.

\textbf{Step 1.} {\it Existence of a solution}. By Lemma \ref{appendix-p-0}, the FBSDE system \eqref{continuation-r-q} admits a solution $(Q,H,R)\in\mathcal H_\alpha\times  S^{2,-}_{\mathbb F}\times L^2_{\mathbb F}$ when $\fp=0$, for any $f\in L^2_{\mathbb F}$. Hence it remains to prove that if for some $\fp\in[0,1]$ the system \eqref{continuation-r-q} admits a solution for any $f\in L^2_{\mathbb F}$, then the same result holds true for $\fp+\fd$ for some small enough constant $\fd$ that is independent of $\fp$. The proof is similar to proof of Lemma \ref{appendix-p-0}.

For any fixed $(Q,R,f)\in \mathcal H_\alpha\times L^2_{\mathbb F}\times L^2_{\mathbb F}$, we introduce the following system:
\begin{equation}\label{continuation-step2}
    \left\{\begin{aligned}
        d\widetilde Q_t=&~\left(-\Lambda^1_t\widetilde R_t-\Lambda^2_t\mathbb E\left[\left.\gamma_t\widetilde Q_t\right|\mathcal F^0_t\right]+\overline f_t\right)\,dt,\\
        -d\widetilde H_t=&~\left(-\Lambda^1_t A_t\widetilde H_t-\Lambda^2_tA_t\mathbb E[\gamma_t\widetilde Q_t|\mathcal F^0_t]+A_t\overline f_t+\fp\Lambda^3_t\mathbb E[\zeta_t\widetilde R_t|\mathcal F^0_t]+\fp\Lambda^5_t\mathbb E[\varrho_t\widetilde Q_t|\mathcal F^0_t]+\overline g_t\right)\,dt,\\
        & \qquad \qquad +\left(f_t+\fd\Lambda^3_t\mathbb E[\zeta_t R_t|\mathcal F^0_t]+\fd\Lambda^5_t\mathbb E[\varrho_tQ_t|\mathcal F^0_t]\right)\,dt-Z_t\,d W_t,\\
        -d\widetilde R_t=&~\left(\Lambda^4_t \widetilde Q_t+\fp\Lambda^3_t\mathbb E[\zeta_t\widetilde R_t|\mathcal F^0_s]+\fd\Lambda^3_t\mathbb E[\zeta_t R_t|\mathcal F^0_s]+\fp\Lambda^5_t\mathbb E[\varrho_t\widetilde Q_t|\mathcal F^0_t]+\fd\Lambda^5_t\mathbb E[\varrho_tQ_t|\mathcal F^0_t] \right. \\ & \qquad \qquad \left.+\overline g_t+f_t\right)\,dt-Z_t\,d W_t,\\
        \widetilde R=&~ A\widetilde Q+\widetilde H,~t\in[0,T),\\
        \widetilde Q_0=&~\chi, ~\widetilde Q_T=0.
    \end{aligned}\right.
\end{equation}
Since $f+\fd\Lambda^3\mathbb E[\zeta R|\mathcal F^0]+\fd\Lambda^5\mathbb E[\varrho Q|\mathcal F^0]\in L^2_{\mathbb F}$,
there exists a solution $(\widetilde Q,\widetilde H,\widetilde R)\in\mathcal H_\alpha\times S^{2,-}_{\mathbb F}\times L^2_{\mathbb F}$ by assumption. This defines a mapping
\begin{equation}\label{mapping-continuation-step2}
    \Phi:(Q,R)\in\mathcal H_\alpha\times L^2_{\mathbb F}\rightarrow(\widetilde Q,\widetilde R)\in\mathcal H_\alpha\times L^2_{\mathbb F}.
\end{equation}

It is sufficient to prove the existence of a fixed point of $\Phi$. To this end, for any $Q,~Q'\in \mathcal H_\alpha$, $R,~R'\in L^2_{\mathbb F}$, by integration by part and using the same arguments leading to the estimate \eqref{p-bar-0-L2-Q},
\begin{equation}\label{estimate-r-q-L2}
\begin{split}
    & \mathbb E\left[\int_0^{T}(\widetilde R_t-\widetilde R'_t)^2\,dt\right]+\mathbb E\left[\int_0^{T}(\widetilde Q_t-\widetilde Q'_t)^2\,dt\right] \\
    \leq & \fd C\mathbb E\left[\int_0^{T}(Q_t-Q'_t)^2\,dt\right]+\fd C\mathbb E\left[\int_0^{T}(R_t-R'_t)^2\,dt\right].
\end{split}
\end{equation}
The preceding estimate allows us to estimate $\widetilde Q$ in terms of $\widetilde R$ as follows 
\begin{equation}\label{estimate-q-S2}
    \begin{split}
         & \mathbb E\left[\sup_{0\leq t\leq T}|\widetilde Q_t-\widetilde Q'_t|^2\right] \\ \leq&~ C\mathbb E\left[\int_0^T|\widetilde R_s-\widetilde R'_s|^2\,ds\right]
        +C\int_0^{T}\mathbb E\left[|\widetilde Q'_s-\widetilde Q_s|^2\right]\,ds\\
        \leq&~ \fd C\mathbb E\left[\int_0^{T}(Q_t-Q'_t)^2\,dt\right]+\fd C\mathbb E\left[\int_0^{T}(R_t-R'_t)^2\,dt\right].
    \end{split}
\end{equation}
By \eqref{estimate-q-S2}, a similar argument as in \eqref{p-bar-0-S2-H} yields the existence of a uniform $C$ such that for any $\epsilon>0$,
\begin{equation}\label{estimate-H-S2}
    \begin{split}
     \mathbb E\left[\sup_{0\leq t\leq T-\epsilon}\left|\widetilde H_t-\widetilde H'_t\right|^2\right]
        \leq&~ C\mathbb E\left[\sup_{0\leq s\leq T}|\widetilde Q_s-\widetilde Q'_s|^2\right]+C\mathbb E\left[\int_0^T|\widetilde R_t-\widetilde R'_t|^2\,dt\right] \\
         & + C\fd\mathbb E\left[\sup_{0\leq s\leq T}|Q_s-Q'_s|^2\right]+C\fd\mathbb E\left[\int_0^T|R_t-R'_t|^2\,dt\right].
    \end{split}
\end{equation}
Now we return to the expression of $\widetilde Q$ in terms of $\widetilde H$, from which we have by \eqref{estimate-q-S2}, \eqref{estimate-H-S2} and the same argument as in \eqref{p-bar-0-Q-S2-T-epsilon} that,
\begin{equation}\label{estimate-q-Halpha}
    \begin{split}
    \mathbb E\left[\sup_{0\leq t\leq T}\left|\frac{\widetilde Q_t-\widetilde Q'_t}{(T-t)^\alpha}\right|^2\right]
    \leq C\fd\|Q-Q'\|^2_{\alpha}+C\fd\mathbb E\left[\int_0^T|R_t-R'_t|^2\,dt\right].
    \end{split}
\end{equation}
By the estimates \eqref{estimate-r-q-L2} and \eqref{estimate-q-Halpha}, when $\fd$ is small enough we have a fixed point which is a solution to \eqref{continuation-r-q} when $\fp$ is replaced by $\fp+\fd$. Iterating the argument finitely often and then taking $f=0$ yields the existence of a solution.

\textbf{Step 2.} {\sl Uniqueness of solutions.}  Let us assume to the contrary that there exist two solutions $(Q,H,R)\in\mathcal H_\alpha\times S^{2,-}_{\mathbb F}\times L^2_{\mathbb F}$ and $(Q',H',R')\in \mathcal H_\alpha\times S^{2,-}_{\mathbb F}\times L^2_{\mathbb F}$ to \eqref{FBSDE-appendix}. As in the proof of Step 1. integration by part for $(Q-Q')(R-R')$ yields,
\begin{equation}\label{L2-0}
    \mathbb E\left[\int_0^T(R_t-R'_t)^2+(Q_t-Q'_t)^2\,dt\right]=0.
\end{equation}
Secondly, by the expression of $(Q-Q')$ in terms of $R-R'$, \eqref{L2-0} yields that
\begin{equation}\label{S2-0-Q}
    \mathbb E\left[\sup_{0\leq t\leq T}|Q_t-Q'_t|^2\right]=0.
\end{equation}
Thirdly, the expression for $(H-H')$, \eqref{L2-0} and \eqref{S2-0-Q} yield that for any $\epsilon>0$
\begin{equation}\label{S2-0-H}
    \mathbb E\left[\sup_{0\leq t\leq T-\epsilon}|H_t-H'_t|^2\right]=0.
\end{equation}
Finally, by the expression for $(Q-Q')$ in terms of $(H-H')$, \eqref{L2-0}, \eqref{S2-0-Q}, \eqref{S2-0-H} and arbitrariness of $\epsilon$ yield that
\begin{equation}
    \mathbb E\left[\sup_{0\leq t\leq T}\left|\frac{Q_t-Q'_t}{(T-t)^\alpha}\right|^2\right]=0.
\end{equation}
\end{proof}
\begin{remark}\label{increase-regularity}
From the proof of Lemma \ref{appendix-p-0} and Theorem \ref{existence-appendix} (see e.g. \eqref{remark-equation-1} and \eqref{remark-equation-2}), we see that for $\overline f\equiv 0$, the regularity of the solution can be increased to $(Q,H)\in\mathcal H_\beta\times\mathcal H_{\varsigma}$, where $\varsigma<\frac{1}{2}\wedge\beta$. This is the case in \cite{FGHP-2018}.
\end{remark}

The following corollary is important for the analysis of our leader-follower game of optimal portfolio liquidation analyzed below. It implies that the follower's optimal response function is linear convex and hence that the leader's control problem is convex.

\begin{corollary}\label{mu-convex-xi0}
The mapping $(\overline f,\overline g)\in S^2_{\mathbb F}\times L^2_{\mathbb F}\rightarrow (Q,H,R)(\overline f, \overline g)\in\mathcal H_\alpha\times S^{2,-}_{\mathbb F}\times L^2_{\mathbb F}$ is well defined and convex.
\end{corollary}
\begin{proof}
By Theorem \ref{existence-appendix}, for each $(\overline f,\overline g)\in S^2_{\mathbb F}\times L^2_{\mathbb F}$, there exists a unique solution $(Q,H,R)$. Thus, the mapping is well defined. Moreover, by the uniqueness again, we have for $\rho\in[0,1]$
\[
    (Q,H,R)(\rho(\overline f,\overline g)+(1-\rho)(\overline f',\overline g'))=\rho (Q,H,R)(\overline f,\overline g)+(1-\rho)(Q,H,R)(\overline f',\overline g').
\]
\end{proof}

Using the same arguments as in the proof of Theorem \ref{existence-appendix} we can also get existence of a unique solution to the ``penalized version'' of \eqref{FBSDE-appendix} where the terminal state constraint on the forward process is replaced by the terminal condition of the backward process $R_T=2nQ_T$. To this end, we introduce the BSDE,
\[
    -dA^n_t=\left(\Lambda^4_t-\Lambda^1_t(A^n_t)^2\right)\,dt-Z_t\,d W_t,\quad A^n_T=2n.
\]
Existence and uniqueness of a solution to this equation has been established in \cite{AJK-2014}. Moreover, for each $t\in[0,T)$,
\begin{equation}\label{convergence-A-n-appendix}
    \lim_{n\rightarrow\infty} A^n_t=A_t, ~\textrm{a.s.}.
\end{equation}
When the terminal state constraint is replaced by the penalty term introduced above, the system \eqref{FBSDE-appendix} translates into the following system:
\begin{equation}\label{FBSDE-appendix-n}
    \left\{\begin{aligned}
         dQ^n_t=&~\left(-\Lambda^1_tR^n_t-\Lambda^2_t\mathbb E\left[\left.\gamma_tQ^n_t\right|\mathcal F^0_t\right]+\overline f^n_t\right)\,dt,\\
        -dH^n_t=&~\left(-\Lambda^1_t A^n_t-\Lambda^2_tA^n_t\mathbb E[\gamma_tQ^n_t|\mathcal F^0_t]+A^n_t\overline f^n_t+\Lambda^3_t\mathbb E[\zeta_tR^n_t|\mathcal F^0_t] \right. \\ & \left. \qquad +\Lambda^5_t\mathbb E[\varrho_tQ^n_t|\mathcal F^0_t]+\overline g^n_t\right)\,dt-Z_t\,d W_t,\\
        -dR^n_t=&~\left(\Lambda^4_t Q^n_t+\Lambda^3_t\mathbb E[\zeta_tR^n_t|\mathcal F^0_s]+\Lambda^5_t\mathbb E[\varrho_tQ^n_t|\mathcal F^0_t]+\overline g^n_t\right)\,dt-Z_t\,d W_t,\\
        Q^n_0=&~\chi, ~H^n_T=0,~R^n_T=2nQ^n_T,
    \end{aligned}\right.
\end{equation}

\begin{corollary}\label{ex-n-appendix}
Assume that for each fixed $n \in \mathbb N$, $(\overline f^n,g^n)\in S^2_{\mathbb F}\times L^2_{\mathbb F}$. Then, for each $n \in \mathbb N$ the FBSDE \eqref{FBSDE-appendix-n} admits a unique solution $(Q^n,H^n,R^n)\in \mathcal H_{\alpha,n}\times S^2_{\mathbb F}\times L^2_{\mathbb F}$,
where
\[
    \mathcal H_{\alpha,n}=\left\{X:\mathbb E\left[\sup_{0\leq t\leq T}\left|\frac{X_t}{(T-t+\frac{1}{n})^\alpha}\right|^2\right]<\infty\right\}.
\]
\end{corollary}
\begin{remark}
Note that in \eqref{FBSDE-appendix-n}, the terminal condition for $H^n$ is $0$ so $H^n$ is defined on $[0,T]$. In \eqref{FBSDE-appendix} the process $H$ is only defined on $[0,T)$, due to to the singularity of the process $A$ at the terminal time.
\end{remark}


\subsection{Convergence}
We now prove an approximation result for the system \eqref{FBSDE-appendix} in terms of the systems \eqref{FBSDE-appendix-n} as $n \to \infty$. The convergence result is established under the additional assumption that for any $0\leq t_1<t_2\leq T$,
\begin{equation}\label{additional-ass-sta}
    e^{-\int_{t_1}^{t_2}\Lambda^1_uA_u\,du}\leq C\frac{T-t_2}{T-t_1}\textrm{  and  }    e^{-\int_{t_1}^{t_2}\Lambda^1_uA^n_u\,du}\leq C\frac{T-t_2+\frac{1}{n}}{T-t_1+\frac{1}{n}}.
\end{equation}
We refer to \cite{FGHP-2018} for sufficient conditions on the model parameters under which this assumption is satisfied.

The proof of the following lemma can be found in \cite[Lemma 4.4]{FGHP-2018}. 

\begin{lemma}\label{boundedness-appendix}
Let $\overline f^n\in S^2_{\mathbb F}$ and $\overline g^n\in L^2_{\mathbb F}$ be two sequences of progressively measurable stochastic processes and $(Q^n,H^n,R^n)$ be the solution to the system \eqref{FBSDE-appendix-n}. If the sequences $\overline f^n$ and $\overline g^n$ are bounded in $S^2_{\mathbb F}$ and $L^2_{\mathbb F}$ uniformly in $n$, respectively, then
\[
    \sup_n\|Q^n\|_{\alpha,n}+\sup_n\|H^n\|_{S^{2,-}}+\sup_n\|R^n\|_{L^2}\leq C\left(\sup_n\|\overline f^n\|_{S^2}+\sup_n\|\overline g^n\|_{L^2}\right)<\infty.
\]
\end{lemma}

\begin{lemma}\label{auxiliary-stability}
Let $\overline f^n$ and $\overline g^n$ be two sequences of stochastic processes satisfying the conditions in Lemma \ref{boundedness-appendix}. Then there exists $\overline f\in L^2_{\mathbb F}$, $\overline g\in L^2_{\mathbb F}$ and a convex combination of a subsequence of $(\overline f^n,\overline g^n)$ converging to $(\overline f,g)$ in $L^\nu$ with $1<\nu<2$, i.e.,
\begin{equation}\label{convergence-fn-gn-Lnu}
    \lim_{N\rightarrow\infty}\mathbb E\left[\int_0^T\left|\frac{1}{N}\sum_{k=1}^N(\overline f^{n_k}_t,\overline g^{n_k}_t)-(\overline f_t,\overline g_t)\right|^\nu\,dt\right]=0.
\end{equation}
\end{lemma}
\begin{proof}
Since the sequence $(\overline f^n,\overline g^n)$ is $L^2$ uniformly bounded, the proof of \cite[Theorem 2.1]{BKOW-2004} 
tells us there exists a subsequence of $(\overline f^n,\overline g^n)$ and a progressively measurable stochastic processes $(\overline f,\overline g)$ such that
\[
    \lim_{N\rightarrow\infty}\frac{1}{N}\sum_{k=1}^N(\overline f^{n_k},\overline g^{n_k})-(\overline f,\overline g)=0,\quad \textrm{a.e. a.s. on }[0,T]\times\Omega.
\]
Fatou's lemma implies that
\[
    \mathbb E\left[\int_0^T|(\overline f_t,\overline g_t)|^2\,dt\right]\leq\liminf_{N\rightarrow\infty}\frac{1}{N}\sum_{k=1}^N\mathbb E\left[\int_0^T|(\overline f^{n_k}_t,\overline g^{n_k}_t)|^2\,dt\right]<\infty.
\]
Thus, Vitali's convergence result implies \eqref{convergence-fn-gn-Lnu}.
\end{proof}

The following theorem proves a convergence result for the FBSDE systems associated with the unconstrained penalized control problems to the system associated with the constrained one. The result is key to our maximum principle for the leader-follower game introduced above.

\begin{theorem}\label{stability-appendix}
Let $(\overline f^n,\overline g^n)$ be the sequence satisfying the conditions in Lemma \ref{auxiliary-stability} and $(\overline f,\overline g)\in L^2_{\mathbb F}\times L^2_{\mathbb F}$ be the limit. Let $(Q^n,H^n,R^n)$ and $(Q,H,R)$ be the solution to \eqref{FBSDE-appendix-n} and \eqref{FBSDE-appendix}, respectively. We further assume the limit $\overline f\in S^2_{\mathbb F}$. Then there exists a convex combination of a subsequence of $\left(\frac{1}{N}\sum_{k=1}^NQ^{n_k},\frac{1}{N}\sum_{k=1}^NH^{n_k},\frac{1}{N}\sum_{k=1}^NR^{n_k}\right)$ converging to $(Q,H,R)$ in $S^\nu_{\mathbb F}\times L^1_{\mathbb F}\times L^\nu_{\mathbb F}$, i.e.,
\begin{equation*}
    \begin{split}
     &~\lim_{N'\rightarrow\infty}\mathbb E\left[\sup_{0\leq t\leq T}\left|\frac{1}{N'}\sum_{j=1}^{N'}\frac{1}{N_j}\sum_{k=1}^{N_j}Q^{n_k}_t-Q_t\right|^\nu\right]=0,\\
     &~\lim_{N'\rightarrow\infty} \mathbb E\left[\int_0^T\left|\frac{1}{N'}\sum_{j=1}^{N'}\frac{1}{N_j}\sum_{k=1}^{N_j}H^{n_k}_t- H_t\right|\,dt\right]=0,\\
     &~\lim_{N'\rightarrow\infty}\mathbb E\left[\int_0^T\left|\frac{1}{N'}\sum_{j=1}^{N'}\frac{1}{N_j}\sum_{k=1}^{N_j}R^{n_k}_t- R_t\right|^\nu\,dt\right]=0.
    \end{split}
\end{equation*}
\end{theorem}
\begin{proof}

The uniform boundedness of $\overline f^n$ and $\overline g^n$ implies the uniform boundedness of $R^n$ in $L^2$ (Lemma \ref{boundedness-appendix}) and the uniform boundedness of $\frac{1}{N}\sum_{k=1}^NR^{n_k}$ in $L^2$. Thus, \cite{BKOW-2004} again yields the existence of a progressively measurable process $\overline R\in L^2_{\mathbb F}$ and a subsequence of $\frac{1}{N}\sum_{k=1}^NR^{n_k}$ such that
\begin{equation}\label{convergence-Rn-Lnu}
    \lim_{N'\rightarrow\infty}\mathbb E\left[\int_0^T\left|\frac{1}{N'}\sum_{j=1}^{N'}\frac{1}{N_j}\sum_{k=1}^{N_j}R^{n_k}_t-\overline R_t\right|^\nu\,dt\right]=0.
\end{equation}
By \eqref{convergence-fn-gn-Lnu}, the convergence of the same convex combination holds for $(\overline f^n,\overline g^n)$:
\begin{equation}\label{convergence-fn-gn-Lnu-2}
    \lim_{N'\rightarrow\infty}\mathbb E\left[\int_0^T\left|\frac{1}{N'}\sum_{j=1}^{N'}\frac{1}{N_j}\sum_{k=1}^{N_j}(\overline f^{n_k}_t,\overline g^{n_k}_t)-(\overline f_t,\overline g_t)\right|^\nu\,dt\right]=0.
\end{equation}
Define $\overline Q$ as the unique solution in $S^2_{\mathbb F}$ to the following mean field SDE in terms of the limits $\overline f$ and $\overline R$:
\begin{equation}\label{def-Q-bar}
    \overline Q_t=\chi+\int_0^t\left(-\Lambda^1_s\overline R_s-\Lambda^2_s\mathbb E[\gamma_s\overline Q_s|\mathcal F^0_s]+\overline f_s\right)\,ds.
\end{equation}
Standard SDE estimates, \eqref{convergence-Rn-Lnu} and \eqref{convergence-fn-gn-Lnu-2} yield,
\begin{equation}\label{convergence-Q-Snu}
   \lim_{N'\rightarrow\infty}\mathbb E\left[\sup_{0\leq t\leq T}\left|\frac{1}{N'}\sum_{j=1}^{N'}\frac{1}{N_j}\sum_{k=1}^{N_j}Q^{n_k}_t-\overline Q_t\right|^\nu\right]=0.
\end{equation}
Now define $\overline H$ in terms of the limits $\overline f$, $\overline R$ and $\overline Q$ as
\begin{equation}\label{def-H-bar}
\begin{split}
    \overline H_t=& \mathbb E\left[\left.\int_t^Te^{-\int_t^s\Lambda^1_u A_u\,du}\left(-\Lambda^2_sA_s\mathbb E[\gamma_s\overline Q_s|\mathcal F^0_s]+A_s\overline f_s+\Lambda^3_s\mathbb E[\zeta_s\overline R_s|\mathcal F^0_s] \right. \right. \right. \\
    & \left. \left. \left. \qquad +\Lambda^5_s\mathbb E[\varrho_s\overline Q_s|\mathcal F^0_s]+\overline g_s\right)\,ds\right|\mathcal F_t\right].
   \end{split}
\end{equation}
Thus, by \eqref{estimate-A}, \eqref{additional-ass-sta} and H\"older inequality,
\begin{align*}
        &~\left|\frac{1}{N'}\sum_{j=1}^{N'}\frac{1}{N_j}\sum_{k=1}^{N_j}H^{n_k}_t-\overline H_t\right|\\
        \leq &~\frac{C}{N'}\sum_{j=1}^{N'}\frac{1}{N_j}\sum_{k=1}^{N_j}\left(\mathbb E\left[\left.\left(\int_t^T\left|e^{-\int_t^s\Lambda^1_uA^{n_k}_u\,du}A^{n_k}_s-e^{-\int_t^s\Lambda^1_uA^{}_u\,du}A^{}_s\right|\,ds\right)^2\right|\mathcal F_t\right]\right)^{\frac{1}{2}}\\
        &~\times\left(\mathbb E\left[\left.\sup_{0\leq s\leq T}|\mathbb E[Q^{n_k}_s|\mathcal F^0_s]|^2+\sup_{0\leq s\leq T}(\overline f^{n_k}_s)^2\right|\mathcal F_t\right]\right)^{\frac{1}{2}}\\
        &~+\frac{C}{N'}\sum_{j=1}^{N'}\frac{1}{N_j}\sum_{k=1}^{N_j}\left(\mathbb E\left[\left.\int_t^T\left|e^{-\int_t^s\Lambda^1_uA^{n_k}_u\,du}-e^{-\int_t^s\Lambda^1_uA^{}_u\,du}\right|^2\,ds\right|\mathcal F_t\right]\right)^{\frac{1}{2}}\left(\mathbb E\left[\left.\int_t^T|\overline g^{n_k}_s|^2\,ds\right|\mathcal F_t\right]\right)^{\frac{1}{2}}\\
        &~+\frac{C}{(T-t)^{\frac{1}{\nu}}}\left(\mathbb E\left[\left.\int_0^T\mathbb E\left[\left.\left|\frac{1}{N'}\sum_{j=1}^{N'}\frac{1}{N_j}\sum_{k=1}^{N_j}Q^{n_k}_s-\overline Q_s\right|^\nu\right|\mathcal F^0_s\right]+\left|\frac{1}{N'}\sum_{j=1}^{N'}\frac{1}{N_j}\sum_{k=1}^{N_j}f^{n_k}_s-\overline f_s\right|^\nu\,ds\right|\mathcal F_t\right]\right)^{\frac{1}{\nu}}\\
        &~+\frac{C}{N'}\sum_{j=1}^{N'}\frac{1}{N_j}\sum_{k=1}^{N_j}\left(\mathbb E\left[\left.\int_t^T\left|e^{-\int_t^s\Lambda^1_uA^{n_k}_u}-e^{-\int_t^s\Lambda^1_uA^{}_u}\right|^2\,ds\right|\mathcal F_t\right]\right)^{\frac{1}{2}}\left(\mathbb E\left[\left.\int_t^T\mathbb E[(R^{n_k}_s)^2+(Q^{n_k}_s)^2|\mathcal F^0_s]\,ds\right|\mathcal F_t\right]\right)^{\frac{1}{2}}\\
        &~+C\mathbb E\left[\left.\int_0^T\mathbb E\left[\left.\left|\frac{1}{N'}\sum_{j=1}^{N'}\frac{1}{N_j}\sum_{k=1}^{N_j}R^{n_k}_s-\overline R_s\right|+\left|\frac{1}{N'}\sum_{j=1}^{N'}\frac{1}{N_j}\sum_{k=1}^{N_j}Q^{n_k}_s-\overline Q_s\right|\right|\mathcal F^0_s\right]\,ds\right|\mathcal F_t\right]\\
        &~+\mathbb E\left[\left.\int_t^T\left|\frac{1}{N'}\sum_{j=1}^{N'}\frac{1}{N_j}\sum_{k=1}^{N_j}\overline g^{n_k}_s-\overline g_s\right|\,ds\right|\mathcal F_t\right].
\end{align*}
Applying H\"older's inequality again along with Doob's maximal inequality, the uniform boundedness of $(Q^n,R^n,\overline f^n,\overline g^n)$ and the dominated convergence theorem we get,
\begin{equation}\label{convergence-H-L1}
   \lim_{N'\rightarrow\infty} \mathbb E\left[\int_0^T\left|\frac{1}{N'}\sum_{j=1}^{N'}\frac{1}{N_j}\sum_{k=1}^{N_j}H^{n_k}_t-\overline H_t\right|\,dt\right]=0.
\end{equation}
Let $\widehat R= A\overline Q+\overline H$. For any $\widetilde T<T$, by \eqref{convergence-Q-Snu} and \eqref{convergence-H-L1} we have
\[
    \lim_{N'\rightarrow\infty}\mathbb E\left[\int_0^{\widetilde T}\left|\frac{1}{N'}\sum_{j=1}^{N'}\frac{1}{N_j}\sum_{k=1}^{N_j}R^{n_k}_t-\widehat R_t\right|\,dt\right]=0.
\]
Thus, \eqref{convergence-Rn-Lnu} implies that for any $\widetilde T<T$,
\[
    \mathbb E\left[\int_0^{\widetilde T}|\widehat R_t-\overline R_t|\,dt\right]=0.
\]
This proves that
\[
    \widehat R=\overline R, \textrm{ a.e. a.s. on }[0,T]\times\Omega.
\]

Thus, the limit $(\overline Q,\overline H,\widehat R)$ satisfies the system \eqref{FBSDE-appendix}. Moreover,
\[
	(\overline Q,\overline H,\widehat R)\in\mathcal H_\alpha\times S^{2,-}_{\mathbb F}\times L^2_{\mathbb F}.
\]
Indeed, since $\overline R\in L^2_{\mathbb F}$ and $\overline R=\widetilde R$ a.e. a.s. on $[0,T]\times\Omega$, we have that $\widehat R\in L^2_{\mathbb F}$. Moreover, \eqref{def-Q-bar} implies that $\overline Q\in S^2_{\mathbb F}$, from which \eqref{def-H-bar} implies $H\in S^{2,-}_{\mathbb F}$ and taking $\widehat R= A\overline Q+\overline H$ into \eqref{def-Q-bar} yields $\overline Q\in\mathcal H_\alpha$. Hence, the uniqueness of solutions in $\mathcal H_\alpha\times S^{2,-}_{\mathbb F}\times L^2_{\mathbb F}$ yields the desired convergence result.
\end{proof}


\section{A MFC problem of optimal portfolio liquidation}\label{section-MKV}

In this section, we solve the single-player portfolio liquidation model with expectations feedback introduced in Section \ref{portfolio1}. We make the following assumption which implies Assumption \ref{ass-ex-sta}.

\begin{ass}\label{ass-MKV}
The process $\widetilde g$ belongs to $L^2_{\mathbb F}$. The progressively measurable stochastic processes $\eta$, $\kappa$ and $\lambda$ are nonnegative and essentially bounded. Moreover, there exists some $\theta'>0$ such that
\[
    \etamin-\frac{\|\kappa\|}{2\theta'}>0,\qquad \lambdamin-\|\kappa\|\theta'>0.
\]
\end{ass}

The trader's objective is to minimize the cost function $J(\cdot)$ introduced in \eqref{cost-MKV} over the set of admissible controls
\[
    \mathcal A_{\mathbb F}(x):=\left\{\xi\in L^2_{{\mathbb F}}([0,T]\times\Omega;\mathbb R):\int_0^T\xi_s\,ds=x\right\}.
\]
A standard stochastic maximum principle suggests the candidate optimal strategy is given by
\begin{equation}\label{MKV-optimal-trading-rate}
    \xi^*_t=\frac{Y_t-\mathbb E[\kappa_tX_t|\mathcal F^0_t]}{2\eta_t}
\end{equation}
where $(X,Y)\in\mathcal H_\alpha\times L^2_{\mathbb F}$ is the unique solution to the FBSDE system \eqref{MKV-FBSDE}. Standard arguments show that $\xi^* \in \mathcal A_{\mathbb F}(x)$. To prove that $\xi^*$ is indeed the unique optimal control, we establish an auxiliary result that substitutes for the lack of convexity of the Hamiltonian for our MFC problem.

\begin{lemma}\label{auxiliary-convexity}
For every $t\in[0,T)$, we have
\begin{equation}
    \begin{split}
        &~\mathbb E\left[\kappa_tX_t\mathbb E[\xi_t|\mathcal F^0_t]+\eta_t\xi^2_t+\lambda_tX^2_t\right]-\mathbb E\left[\kappa_tX^*_t\mathbb E[\xi^*_t|\mathcal F^0_t]+\eta_t(\xi^*_t)^2+\lambda_t(X^*_t)^2\right]\\
        \geq&~\mathbb E\left[\left(\mathbb E[\kappa_tX^*_t|\mathcal F^0_t]+2\eta_t\xi^*_t\right)(\xi_t-\xi^*_t)+2\lambda_tX^*_t(X_t-X^*_t)+\kappa_t(X_t-X^*_t)\mathbb E[\xi^*_t|\mathcal F^0_t]\right].
    \end{split}
\end{equation}
Moreover, the above inequality becomes an equality if and only if $\xi_t= \xi^*_t$ a.s..
\end{lemma}
\begin{proof}
To prove \eqref{auxiliary-convexity}, it is equivalent to show
\[
    \mathbb E\left[\eta_t(\xi_t-\xi^*_t)^2+\lambda_t(X_t-X^*_t)^2+\mathbb E[(\xi_t-\xi^*_t)|\mathcal F^0_t]\mathbb E[\kappa_t(X_t-X^*_t)|\mathcal F^0_t]\right]\geq 0.
\]
Note that
\begin{equation*}
    \begin{split}
        &~|\mathbb E\left[\mathbb E[(\xi_t-\xi^*_t)|\mathcal F^0_t]\mathbb E[\kappa_t(X_t-X^*_t)|\mathcal F^0_t]\right]|\\
        \leq&~\|\kappa\|\mathbb E\left[\mathbb E[|\xi_t-\xi^*_t||\mathcal F^0_t]\mathbb E[|X_t-X^*_t||\mathcal F^0_t]\right]\\
        \leq&~\frac{\|\kappa\|}{2\theta}\mathbb E\left[\left(\mathbb E[|\xi_t-\xi^*_t||\mathcal F^0_t]\right)^2\right]+\frac{\|\kappa\|\theta}{2}\mathbb E\left[\left(\mathbb E[|X_t-X^*_t||\mathcal F^0_t]\right)^2\right].
    \end{split}
\end{equation*}
Thus,
\begin{equation*}
    \begin{split}
    &~\mathbb E\left[\eta_t(\xi_t-\xi^*_t)^2+\lambda_t(X_t-X^*_t)^2+\mathbb E[(\xi_t-\xi^*_t)|\mathcal F^0_t]\mathbb E[\kappa_t(X_t-X^*_t)|\mathcal F^0_t]\right]\\
    \geq&~\mathbb E\left[\left(\etamin-\frac{\|\kappa\|}{2\theta}\right)(\xi_t-\xi^*_t)^2+\left(\lambdamin-\frac{\|\kappa\|\theta}{2}\right)(X_t-X^*_t)^2-\|\kappa\|\mathbb E[|\xi_t-\xi^*_t||\mathcal F^0_t]\mathbb E[|X_t-X^*_t||\mathcal F^0_t]\right]\\
    &~+\frac{\|\kappa\|}{2\theta}\mathbb E\left[(\xi_t-\xi^*_t)^2\right]+\frac{\|\kappa\|\theta}{2}\mathbb E\left[(X_t-X^*_t)^2\right]\\
    \geq&~\mathbb E\left[\left(\etamin-\frac{\|\kappa\|}{2\theta}\right)(\xi_t-\xi^*_t)^2+\left(\lambdamin-\frac{\|\kappa\|\theta}{2}\right)(X_t-X^*_t)^2-\|\kappa\|\mathbb E[|\xi_t-\xi^*_t||\mathcal F^0_t]\mathbb E[|X_t-X^*_t||\mathcal F^0_t]\right]\\
    &~+\frac{\|\kappa\|}{2\theta}\mathbb E\left[(\mathbb E[|\xi_t-\xi^*_t||\mathcal F^0_t])^2\right]+\frac{\|\kappa\|\theta}{2}\mathbb E\left[(\mathbb E[|X_t-X^*_t||\mathcal F^0_t])^2\right]\\
    \geq&~\mathbb E\left[\left(\etamin-\frac{\|\kappa\|}{2\theta}\right)(\xi_t-\xi^*_t)^2+\left(\lambdamin-\frac{\|\kappa\|\theta}{2}\right)(X_t-X^*_t)^2\right]\\
    \geq&~0.
    \end{split}
\end{equation*}
The second claim is obvious from the above estimate.
\end{proof}

We are now ready to state and prove the main result of this section.

\begin{theorem}
	Under Assumption \ref{ass-MKV} the process $\xi^*$ defined in \eqref{MKV-optimal-trading-rate} is the unique optimal control to the MFC problem \eqref{cost-MKV}-\eqref{state-MFC}.
\end{theorem}
\begin{proof}
To prove the optimality of the candidate strategy $\xi^*$ we fix an arbitrary control $\xi \in \mathcal A_{\mathbb F}(x)$ and denote by $X^*$ and $X$ the corresponding state processes. For any $\epsilon>0$, it follows from Lemma \ref{auxiliary-convexity} that
\begin{equation}\label{SMP-1}
    \begin{split}
        &~\mathbb E\left[\int_0^{T-\epsilon}\kappa_tX_t\mathbb E[\xi_t|\mathcal F^0_t]+\widetilde g_tX_t+\eta_t\xi^2_t+\lambda_tX^2_t\,dt\right]\\
        &~-\mathbb E\left[\int_0^{T-\epsilon}\kappa_tX^*_t\mathbb E[\xi^*_t|\mathcal F^0_t]+\widetilde g_tX^*_t+\eta_t(\xi^*_t)^2+\lambda_t(X^*_t)^2\,dt\right]\\
        \geq&~\mathbb E\left[\int_0^{T-\epsilon}\left(\mathbb E[\kappa_tX^*_t|\mathcal F^0_t]+2\eta_t\xi^*_t\right)(\xi_t-\xi^*_t)+(2\lambda_tX^*_t+\kappa_t\mathbb E[\xi^*_t|\mathcal F^0_t]+\widetilde g_t)(X_t-X^*_t)\,dt\right].
    \end{split}
\end{equation}
Integration by part yields,
\begin{equation}\label{SMP-2}
    \begin{split}
        &\mathbb E\left[Y_{T-\epsilon}(X_{T-\epsilon}-X^*_{T-\epsilon})\right]\\
            =&-\mathbb E\left[\int_0^{T-\epsilon}Y_t(\xi_t-\xi^*_t)\,dt\right]-\mathbb E\left[\int_0^{T-\epsilon}(X_t-X^*_t)\left(\kappa_t\mathbb E\left[\left.\frac{Y^*_t}{2\eta_t}\right|\mathcal F^0_t\right] \right. \right.  \\
        & \left. \left. \qquad -\kappa_t\mathbb E\left[\left.\frac{1}{2\eta_t}\right|\mathcal F^0_t\right]\mathbb E[\kappa_tX^*_t|\mathcal F^0_t]+2\lambda_tX^*_t+\widetilde g_t\right)\,dt\right]\\
        =&~-\mathbb E\left[\int_0^{T-\epsilon}Y_t(\xi_t-\xi^*_t)\,dt\right]-\mathbb E\left[\int_0^{T-\epsilon}(X_t-X^*_t)\left(\kappa_t\mathbb E\left[\xi^*_t|\mathcal F^0_t\right]+2\lambda_tX^*_t+\widetilde g_t\right)\,dt\right].
    \end{split}
\end{equation}
Putting \eqref{SMP-2} into \eqref{SMP-1}, we have
\begin{equation}\label{SMP-3}
    \begin{split}
        &~\mathbb E\left[\int_0^{T-\epsilon}\kappa_tX_t\mathbb E[\xi_t|\mathcal F^0_t]+\widetilde g_tX_t+\eta_t\xi^2_t+\lambda_tX^2_t\,dt\right]\\
        &~-\mathbb E\left[\int_0^{T-\epsilon}\kappa_tX^*_t\mathbb E[\xi^*_t|\mathcal F^0_t]+\widetilde g_tX^*_t+\eta_t(\xi^*_t)^2+\lambda_t(X^*_t)^2\,dt\right]\\
        &~+\mathbb E\left[Y_{T-\epsilon}(X_{T-\epsilon}-X^*_{T-\epsilon})\right]\\
        \geq&~\mathbb E\left[\int_0^{T-\epsilon}\left(\mathbb E[\kappa_tX^*_t|\mathcal F^0_t]+2\eta_t\xi^*_t-Y_t\right)(\xi_t-\xi^*_t)\,dt\right]=0.
    \end{split}
\end{equation}
Letting $\epsilon\rightarrow 0$, a similar argument as the proof of \cite[Theorem 2.9]{FGHP-2018} yields that
\begin{equation*}
    \begin{split}
    &~\lim_{\epsilon\rightarrow 0}\mathbb E[Y_{T-\epsilon}(X_{T-\epsilon}-X^*_{T-\epsilon})]=0.
    \end{split}
\end{equation*}
Thus, \eqref{SMP-3} implies $$J(\xi)\geq J(\xi^*).$$

In order to prove the uniqueness of optimal controls, let $\xi'$ be another optimal control. Then, \eqref{SMP-3} yields
 \begin{equation*}
    \begin{split}
        0=&~\mathbb E\left[\int_0^{T}\kappa_tX_t\mathbb E[\xi'_t|\mathcal F^0_t]+\widetilde g_tX_t'+\eta_t(\xi'_t)^2+\lambda_t(X'_t)^2\,dt\right] \\
        & \qquad -\mathbb E\left[\int_0^{T}\kappa_tX^*_t\mathbb E[\xi^*_t|\mathcal F^0_t]+\widetilde g_tX^*_t+\eta_t(\xi^*_t)^2+\lambda_t(X^*_t)^2\,dt\right]\\
        \geq&~\mathbb E\left[\int_0^{T}\left(\mathbb E[\kappa_tX^*_t|\mathcal F^0_t]+2\eta_t\xi^*_t-Y_t\right)(\xi'_t-\xi^*_t)\,dt\right]=0.
    \end{split}
\end{equation*}
Thus, \eqref{SMP-1} holds with an equality. The second claim in Lemma \ref{auxiliary-convexity} implies the uniqueness.
\end{proof}


\section{A Stackelberg game of optimal portfolio liquidation}\label{section-MFG}

In this section, we solve the Stackelberg game of optimal portfolio liquidation introduced in Section \ref{portfolio2} above. We make the following assumption which implies Assumption \ref{ass-ex-sta} and Assumption \eqref{additional-ass-sta}.

\begin{ass}\label{assumption-MFG}
\begin{itemize}
    \item[(1)] The processes $\widetilde\kappa^0$, $\kappa$, $\eta$, $1/\eta$ and $\lambda$ belong to $L^{\infty}_{\mathbb F}([0,T]\times\Omega;[0,\infty))$.
    \item[(2)] The processes $\overline\kappa^0$, $\kappa^0$, $\eta^0$, $1/\eta^0$ and $\lambda^0$ belong to $L^{\infty}_{\mathbb F^0}([0,T]\times\Omega;[0,\infty))$.
    \item[(3)] For some positive constants $\theta'$, $\theta$ and $\overline\theta$,
    \[
        \etamin-\frac{\|\kappa\|}{2\theta'}>0,\quad \lambdamin-\|\kappa\|\theta'>0.
    \]
    and
     \[
        \quad\etamin^0-\frac{\|\kappa^0\|}{2\theta}>0,\quad \lambdamin^0-\frac{\|\kappa^0\|\theta}{2}-\frac{\|\overline\kappa^0\|\overline\theta}{2}>0,\quad \overline\lambdamin-\frac{\|\overline\kappa^0\|}{2\overline\theta}>0.
    \]
    \item[(4)] For any $0\leq s<t\leq T$,
    \[
        e^{-\int_s^t\frac{A_u}{2\eta_u}\,du}\leq C\left(\frac{T-t}{T-s}\right)
    \]
    and
    \[
        e^{-\int_s^t\frac{A^n_u}{2\eta_u}\,du}\leq C\left(\frac{T-t+\frac{1}{n}}{T-s+\frac{1}{n}}\right).
    \]
\end{itemize}
\end{ass}

The problem of the Stackelberg leader is to minimize the cost functional \eqref{leader-cost} over the set of admissible controls
\[
    \mathcal A_{\mathbb F^0}(x^0):=\left\{\xi^0\in L^2_{\mathbb F^0}([0,T]\times\Omega;\mathbb R):\int_0^T\xi^0_s\,ds=x^0\right\}.
\]

The follower's optimal response function is given by
\begin{equation}\label{followr-optimal-strategy}
    \overline\xi_t:=\overline\xi_t(\xi^0):=\frac{Y_t(\xi^0) - \mathbb E[\kappa_tX_t(\xi^0)|\mathcal F^0_t]}{2\eta_t},
\end{equation}
where $(X,Y)$ is the solution to \eqref{MKV-FBSDE} with $\widetilde g=\widetilde\kappa^0\xi^0$. We will occasionally drop the dependence on $\xi^0$ if there is no confusion.
Under Assumption \ref{assumption-MFG} the solution $(X,Y)$ enjoys better regularity properties, due to Remark \ref{increase-regularity} and the estimate \eqref{estimate-A}.
\begin{corollary}\label{existence-X-Y-B}
Under Assumption \ref{assumption-MFG}, the solution to \eqref{MKV-FBSDE} belongs to $\mathcal H_{1}\times S^2_{\mathbb F}$. Moreover, $Y=AX+B$ with $B\in \mathcal H_{\varsigma}$.
\end{corollary}

In the next section we first prove that the leader's problem has a unique solution if the terminal state constraints are replaced by finite penalty terms and establish a necessary maximum principle for the penalized problem.
Subsequently we prove the convergence of the state and adjoint equations of the penalized problems as the degree of penalization tends to infinity.


\subsection{The penalized problem: existence and maximum principle}

The penalized optimization problem is obtained by replacing the terminal state constraint on the leader's and follower's state process by a finite penalty term. The leader's problem consists in minimizing the cost functional
\begin{equation}\label{cost-p-n}
   J^{0,n}(\xi^0)
    :=\mathbb E\left[\int_0^T\overline\kappa^0_s\mathbb E[\overline\xi^n_s|\mathcal F^0_s]X^0_s+\kappa^0_s\xi^0_sX^0_s+\eta^0_s(\xi^0_s)^2+\lambda^0_s(X^0_s)^2+\overline\lambda_s(\mathbb E[\overline\xi^n_s|\mathcal F^0_s])^2\,ds+n(X^{0}_T)^2\right]
    \end{equation}
over all controls $\xi^0\in L^2_{\mathbb F^0}$ subject to the state dynamics
\begin{equation}\label{state-p-n}
    \left\{\begin{aligned}
    dX^0_t=&~-\xi^0_t\,dt,\\
    dX_t=&~-\frac{Y_t-\mathbb E[\kappa_tX_t|\mathcal F^0_t]}{2\eta_t}\,dt,\\
    -dY_t=&~\left(\kappa_t\mathbb E\left[\left.\frac{Y_t}{2\eta_t}\right|\mathcal F^0_t\right]-\kappa_t\mathbb E\left[\left.\frac{1}{2\eta_t}\right|\mathcal F^0_t\right]\mathbb E[\kappa_tX_t|\mathcal F_t]+2\lambda_tX_t+\widetilde\kappa^0_t\xi^0_t\right)\,dt-Z_t\,d W_t,\\
    X_0=&~x,~X^0_0=x^0,~Y_T=2nX_T,
    \end{aligned}\right.
\end{equation}
where the optimal response for the penalized follower $\overline\xi^n$ is defined as follows in terms of $(X,Y)$ in \eqref{state-p-n}
\[
    \overline\xi^n:=\frac{Y-\mathbb E[\kappa X|\mathcal F^0]}{2\eta}.
\]
We are now going to show that the penalized optimization problem has a unique solution. Similar arguments could be used to prove the existence of an optimal control for the original problem. They would not, however, give us an open-loop characterization of the optimal control.

\begin{theorem}\label{p-existence}
For each $n \in \mathbb N$, the penalized optimization problem \eqref{cost-p-n}-\eqref{state-p-n} admits a unique optimal control in $L^2_{\mathbb F^0}$.
\end{theorem}
\begin{proof}
In view of Corollary \ref{ex-n-appendix} the systems \eqref{state-p-n} is well-posed for each fixed $\xi^0\in L^2_{\mathbb F^0}$. The representation of the cost functional
\begin{equation*}
    \begin{split}
    & ~J^{0,n}(\xi^0)\\
    =&~\mathbb E\left[\int_0^T\frac{\overline\kappa^0_t}{2}\left(\sqrt{\overline\theta} X^0_t+\frac{\mathbb E\left[\left.\overline\xi_t^n\right|\mathcal F^0_t\right]}{\sqrt{\overline\theta}}\right)^2+\frac{\kappa^0_t}{2}\left(\sqrt\theta X^0_t+\frac{\xi^0_t}{\sqrt\theta}\right)^2+\left(\lambda^0_t-\frac{\overline\kappa^0_t\overline\theta}{2}-\frac{\kappa^0_t\theta}{2}\right)(X^0_t)^2\right.\\
    &~\left.+\left(\eta^0_t-\frac{\kappa^0_t}{2\theta}\right)(\xi^0_t)^2+\left(\overline\lambda_t-\frac{\overline\kappa^0_t}{2\overline\theta}\right)\left(\mathbb E\left[\left.\overline\xi_t^n\right|\mathcal F^0_t\right]\right)^2\,dt+n(X^0_T)^2\right]
    \end{split}
\end{equation*}
along with Corollary \ref{mu-convex-xi0} and Assumption \ref{assumption-MFG} shows that $J^{0,n}$ is strictly convex. Uniqueness of the optimal strategy follows.

Let $J^* = \inf_{\xi^0 \in L^2_{\mathbb F^0}} J^{0,n}(\xi^0)$. Then $J^* < \infty$ because $J^{0,n}(x^0/T)$ is bounded. Let $\{\xi^{0,n,m}\}\subseteq L^2_{\mathbb F^0}$ be a sequence such that
\[
    \lim_{m\rightarrow\infty}J^{0,n}(\xi^{0,n,m})=J^*.
\]
By Assumption \ref{assumption-MFG} this implies,
\begin{equation}\label{penalized-uniform-bound}
    \sup_m E\left[\int_0^T(\xi^{0,n,m}_s)^2\,ds\right]<C.
\end{equation}
Thus, Lemma \ref{auxiliary-stability} implies the existence of some $\xi^{0,n,*}\in L^2_{\mathbb F^0}$ such that
\begin{equation}\label{penalized-L-nu-convergence-xi0}
    \lim_{\overline N\rightarrow\infty}\mathbb E\left[\int_0^T\left|\overline\xi^{0,n,N}_t-\xi^{0,n,*}_t\right|^\nu\,dt\right]=0,~1<\nu<2,
\end{equation}
where
\[
    \overline\xi^{0,n,N}=\frac{1}{N}\sum_{k=1}^{N}\xi^{0,n,m_k}.
\]
Let $(X^{0,n,*},X^{n,*},Y^{n,*})$ be the solution to \eqref{state-p-n} associated with $\xi^{0,n,*}$. Then the same argument as in the proof of Theorem \ref{stability-appendix} implies,
\[
    \lim_{N\rightarrow\infty}\mathbb E\left[\int_0^T\left|\frac{1}{\overline N}\sum_{j=1}^{\overline N}\frac{1}{N_j}\sum_{k=1}^{N_j}(X^{n,m_k}_t,Y^{n,m_k}_t)-(X^{n,*}_t,Y^{n,*}_t)\right|^\nu\,dt\right]=0,~1<\nu<2.
\]
Moreover, \eqref{penalized-L-nu-convergence-xi0} yields,
\[
    \lim_{\overline N\rightarrow\infty}\mathbb E\left[\sup_{0\leq t\leq T}\left|\frac{1}{\overline N}\sum_{j=1}^{\overline N}\frac{1}{N_j}\sum_{k=1}^{N_j}X^{0,n,m_k}_t-X^{0,n,*}_t\right|^\nu\,dt\right]= 0.
\]

 Thus, Fatou's lemma and the convexity of $J^{0,n}$ imply that
\begin{equation*}
    \begin{split}
       J^{0,n}(\xi^{0,n,*})\leq \liminf_{\overline N\rightarrow\infty}J^{0,n}\left(\frac{1}{\overline N}\sum_{j=1}^{\overline N}\overline\xi^{0,n,N_j}\right)\leq\liminf_{\overline N\rightarrow\infty}\frac{1}{\overline N}\sum_{j=1}^{\overline N}\frac{1}{N_j}\sum_{k=1}^{N_j}J^{0,n}\left(\xi^{0,n,m_k}\right) =J^*.
    \end{split}
\end{equation*}
\end{proof}

From now on, we denote by $\xi^{0,n,*}$ the unique optimal control for the penalized optimization \eqref{cost-p-n}-\eqref{state-p-n}. The following theorem provides a characterization of $\xi^{0,n,*}$.

\begin{theorem}[Necessary maximum principle]\label{p-max-principle}
The optimal control $\xi^{0,n,*}$ admits the following representation:
\begin{equation}\label{characterization-xi0-n}
    \xi^{0,n,*}_t=\frac{p^n_t+\mathbb E[\widetilde\kappa^0_tq^n_t|\mathcal F^0_t]-\kappa^0_tX^{0,n,*}_t}{2\eta^0_t},\quad \textrm{a.e. a.s. on }[0,T]\times\Omega,
\end{equation}
where $X^{0,n,*}$, $p^n$ and $q^n$ satisfy the following FBSDE system:
\begin{equation}\label{p-optimal-system}
    \left\{\begin{aligned}
        dX^{0,n,*}_t=&~-\xi^{0,n,*}_t\,dt,\\
        dX^{n,*}_t=&~-\xi^{n,*}_t\,dt,\\
        -dY^{n,*}_t=&~\left(\kappa_t\mathbb E\left[\xi^{n,*}_t|\mathcal F^0_t\right]+2\lambda_tX^{n,*}_t+\widetilde\kappa^0_t\xi^{0,n,*}_t\right)\,dt-Z_t\,d W_t,\\
        -dp^n_t=&~\left(\overline\kappa^0_t\mathbb E\left[\xi^{n,*}_t|\mathcal F^0_t\right]+\kappa^0_t\xi^{0,n,*}_t+2\lambda^0_tX^{0,n,*}_t\right)\,dt-Z_t\,dW^0_t,\\
        -dq^n_t=&~\left(-\frac{r^n_t}{2\eta_t}-\mathbb E\left[\kappa_tq^n_t|\mathcal F^0_t\right]\frac{1}{2\eta_t}+\overline f^n_t\right)\,dt,\\
        -dr^n_t=&~\left(-2\lambda_tq^n_t+\kappa_t\mathbb E\left[\left.\frac{r_t}{2\eta_t}\right|\mathcal F^0_t\right]+\kappa_t\mathbb E\left[\left.\frac{1}{2\eta_t}\right|\mathcal F^0_t\right]\mathbb E[\kappa_tq^n_t|\mathcal F^0_t]+\overline g^n_t\right)\,dt-Z_t\,d W_t,\\
        X^0_0=&~x^0,~X_0=x,~Y^{n,*}_T=2nX^{n,*}_T,~p^n_T=2nX^{0,n,*}_T,~r^n_T=-2nq^n_T,~q^n_0=0,
    \end{aligned}\right.
\end{equation}
where
    \begin{equation}\label{xi-n-star}
        \xi^{n,*}_t:=\frac{Y^{n,*}_t-\mathbb E[\kappa_tX^{n,*}_t|\mathcal F^0_t]}{2\eta_t},
    \end{equation}
\begin{equation}\label{input-f-bar-n}
    \overline f^n_t:=\frac{\overline\kappa^0_tX^{0,n,*}_t}{2\eta_t}+\frac{\overline\lambda_t}{\eta_t}\mathbb E\left[\xi^{n,*}_t|\mathcal F^0_t\right],
\end{equation}
and
\begin{equation}\label{input-g-bar-n}
    \overline g^n_t:=-\kappa_t\mathbb E\left[\left.\frac{1}{2\eta_t}\right|\mathcal F^0_t\right]\overline\kappa^0_tX^{0,n,*}_t-2\overline\lambda_t\kappa_t\mathbb E\left[\left.\frac{1}{2\eta_t}\right|\mathcal F^0_t\right]\mathbb E\left[\xi^{n,*}_t|\mathcal F^0_t\right].
\end{equation}

\end{theorem}
\begin{proof}
A unique optimal control $\xi^{0,n,*}$ exists, due to Theorem \ref{p-existence}. It is to be viewed as an exogenous input to the FBSDE system \eqref{p-optimal-system}. Thus, the system $(X^{n,*},Y^{n,*})$ is a special case of \eqref{FBSDE-appendix-n} by taking \eqref{xi-n-star} into account. Corollary \ref{ex-n-appendix} implies that the system is well-posed. Considering $\overline f^n$ and $\overline g^n$ as inputs, the system $(q^n,r^n)$ is well-posed, again due to Corollary \ref{ex-n-appendix}. The characterization \eqref{characterization-xi0-n} is then a direct result of stochastic maximum principle for control of FBSDE with partial information; cf~\cite{OS-2009}.
\end{proof}

The ansatz $p^n=\overline A^nX^{0,n,*}+\overline p^n$ shows that the equation for $p^n$ could be dropped from the above system. It yields the following BSDEs for the processes $\overline A^n$ and $\overline p^n$ that will be used in the next subsection:
\begin{equation}\label{A-n-overline}
    \left\{\begin{aligned}
    -d\overline A^n_t=&~\left(-\frac{(\overline A^n_t)^2}{2\eta^0_t}+\frac{\kappa^0_t\overline A^n_t}{2\eta^0_t}+2\lambda^0_t\right)\,dt-Z_t^{\overline A^n}\,d W^0_t,\\
    \overline A^n_T=&~2n
    \end{aligned}\right.
\end{equation}
and
\begin{equation}\label{pn-bar}
    \left\{\begin{aligned}
    -d\overline p^n_t=&~\left(-\frac{\overline A^n_t\overline p^n_t}{2\eta^0_t}-\frac{\overline A^n_t\mathbb E[\widetilde\kappa^0_tq^n_t|\mathcal F^0_t]}{2\eta^0_t}+\kappa^0_t\xi^{0,n,*}_t+\overline\kappa^0_tE\left[\xi^{n,*}_t|\mathcal F^0_t\right]\right)\,dt-Z^{\overline p^n}_t\,dW^0_t,\\
    \overline p^n_T=&~0.
    \end{aligned}\right.
\end{equation}


\subsection{The optimal solution to the Stackelberg game}

Let us recall that $\xi^{0,n,*}$ denotes the leader's optimal control for penalized optimization with index $n \in \mathbb N$. The  uniform boundedness of  $J^{0,n}(x^0/T)$ in $n \in \mathbb N$ implies,
 \begin{equation}\label{uniform-boundedness-xi0n}
    \sup_n\mathbb E\left[\int_0^T\left|\xi^{0,n,*}_t\right|^2\,dt+n(X^{0,n,*}_T)^2\right]<\infty.
 \end{equation}
  Thus, the same arguments as in the proof of Lemma  \ref{auxiliary-stability} yield the existence of a progressively measurable process
\begin{equation}\label{xi0-L2}
    \xi^{0,*}\in L^2_{\mathbb F^0}(\Omega\times[0,T];\mathbb R)
\end{equation}
such that
\begin{equation}\label{app-xi0-2}
    \lim_{N\rightarrow\infty}\mathbb E\left[\int_0^T\left|\frac{1}{N}\sum_{k=1}^N\xi^{0,n_k,*}_t-\xi^{0,*}_t\right|^\nu\,dt\right]=0,~1<\nu<2.
\end{equation}

Our goal is to prove that $\xi^{0,*}$ is the leader's unique optimal strategy in the original state-constrained Stackelberg game. To this end, we first establish a representation of $\xi^{0,*}$ in terms of the solution to the system \eqref{state-leader}, \eqref{leader-adjoint-p} and \eqref{leader-adjoint-q-r} by proving that the solutions to the system of state and adjoint equations \eqref{p-optimal-system} for the unconstrained penalized MFC problem Cesaro converge to the solutions to the systems \eqref{MKV-FBSDE}, \eqref{state-leader}, \eqref{leader-adjoint-p} and \eqref{leader-adjoint-q-r}. From this, we then deduce a sufficient maximum principle for the leader's MFC problem from which we conclude the optimality of the candidate strategy $\xi^{0,*}$.


\subsubsection{Approximation}

With the limit $\xi^{0,*}$ at hand, we can consider the FBSDE system \eqref{MKV-FBSDE}, \eqref{state-leader}, \eqref{leader-adjoint-p} and \eqref{leader-adjoint-q-r} with $\xi^0$ replaced by $\xi^{0,*}$. The system \eqref{MKV-FBSDE} for $(X^{*},Y^*)$ is well-posed, due to Corollary \ref{existence-X-Y-B}. The system for $(q,r)$ is well-posed, due to the following corollary.

\begin{corollary}\label{existence-q-r-D}
If we take $\chi=x^0$, $\Lambda^1=\Lambda^2=\zeta=1/2\eta$, $\gamma=\Lambda^3=\varrho=\kappa$, $\Lambda^4=2\lambda$, $\Lambda^5=\kappa\mathbb E\left[\left.\frac{1}{2\eta}\right|\mathcal F^0\right]$, $Q=-q$,
\begin{equation}\label{input-f-bar}
    \overline f=\frac{\kappa^0X^{0,*}}{2\eta}+\frac{\overline\lambda}{\eta}\mathbb E\left[\xi^*|\mathcal F^0\right]
\end{equation}
and
\begin{equation}\label{input-g-bar}
    \overline g=-\kappa\mathbb E\left[\left.\frac{1}{2\eta}\right|\mathcal F^0\right]\overline\kappa^0X^{0,*}-2\overline\lambda\kappa\mathbb E\left[\left.\frac{1}{2\eta}\right|\mathcal F^0\right]\mathbb E\left[\xi^*|\mathcal F^0\right],
\end{equation}
where
\begin{equation}\label{optimal-strategy-follower-2}
    \xi^*:=\frac{Y^{*}}{2\eta}-\frac{1}{2\eta}\mathbb E[\kappa_tX^{*}|\mathcal F^0].
\end{equation}
Then the system \eqref{FBSDE-general} reduces \eqref{leader-adjoint-q-r}. Hence, existence and uniqueness of a solution holds for \eqref{leader-adjoint-q-r}. Moreover, $r=-Aq+D$ with $D\in S^{2,-}_{\mathbb F}$.
\end{corollary}
We now introduce two BSDEs that we expect to be the limits to the equations \eqref{A-n-overline} and \eqref{pn-bar}:
\begin{equation}\label{limit-A-bar}
    \left\{\begin{aligned}
    -d\overline A_t=&~\left(-\frac{\overline A^2_t}{2\eta^0_t}+\frac{\kappa^0_t\overline A_t}{2\eta^0_t}+2\lambda^0_t\right)\,dt-Z_t\,dW^0_t\\
    \lim_{t\nearrow T}\overline A_t=&~\infty,
    \end{aligned}\right.
\end{equation}
and
\begin{equation}\label{limit-p-bar}
    \left\{\begin{aligned}
    -d\overline p_t=&~\left(-\frac{\overline A_t\overline p_t}{2\eta^0_t}-\frac{\overline A_t\mathbb E[\widetilde\kappa^0_tq_t|\mathcal F^0_t]}{2\eta^0_t}+\kappa^0_t\xi^{0,*}_t+\overline\kappa^0_tE\left[\xi^{*}_t|\mathcal F^0_t\right]\right)\,dt-Z^{\overline p}_t\,d W^0_t,\\
    \overline p_T=&~0.
    \end{aligned}\right.
\end{equation}

where $\xi^*$ and $\xi^{0,*}$ are defined in \eqref{optimal-strategy-follower-2} and \eqref{xi0-L2}, respectively. The following lemma confirms our guess. It shows that the solutions to the FBSDE system \eqref{p-optimal-system} converge to the solutions to the FBSDE systems \eqref{MKV-FBSDE}, \eqref{state-leader}, \eqref{leader-adjoint-q-r} and \eqref{limit-p-bar} in the same sense as the optimal solutions to the unconstrained penalized problems converge to the candidate solution of the constrained problem.

\begin{lemma}\label{app}
For $1<\nu<2$, it holds that
\begin{equation}\label{app-X0}
    \lim_{N\rightarrow\infty}\mathbb E\left[\sup_{0\leq t\leq T}\left|\frac{1}{N}\sum_{k=1}^NX^{0,n_k,*}_t-X^{0,*}_t\right|^\nu\right]=0,
\end{equation}
\begin{equation}\label{app-X}
    \lim_{\overline N\rightarrow\infty}\mathbb E\left[\int_0^T\left|\frac{1}{\overline N}\sum_{j=1}^{\overline N}\frac{1}{N_j}\sum_{k=1}^{N_j}X^{n_k,*}_t-X^*_t\right|^\nu\,dt\right]=0,
\end{equation}
\begin{equation}\label{app-Y}
    \lim_{\overline N\rightarrow\infty}\mathbb E\left[\int_0^T\left|\frac{1}{\overline N}\sum_{j=1}^{\overline N}\frac{1}{N_j}\sum_{k=1}^{N_j}Y^{n_k,*}_t-Y^*_t\right|^\nu\,dt\right]=0,
\end{equation}
\begin{equation}\label{app-q}
    \lim_{\widetilde N\rightarrow\infty}\mathbb E\left[\sup_{0\leq t\leq T}\left|\frac{1}{\widetilde N}\sum_{i=1}^{\widetilde N}\frac{1}{\overline N_i}\sum_{j=1}^{\overline N_i}\frac{1}{N_j}\sum_{k=1}^{N_j}q^{n_k}_t-q_t\right|^\nu\right]=0,
\end{equation}
\begin{equation}\label{app-r}
    \lim_{\widetilde N\rightarrow\infty}\mathbb E\left[\int_0^T\left|\frac{1}{\widetilde N}\sum_{i=1}^{\widetilde N}\frac{1}{\overline N_i}\sum_{j=1}^{\overline N_i}\frac{1}{N_j}\sum_{k=1}^{N_j}r^{n_k}_t-r_t\right|^\nu\,dt\right]=0,
\end{equation}
\begin{equation}\label{app-p-bar}
    \lim_{\widetilde N\rightarrow\infty}\mathbb E\left[\int_0^T\left|\frac{1}{\widetilde N}\sum_{i=1}^{\widetilde N}\frac{1}{\overline N_i}\sum_{j=1}^{\overline N_i}\frac{1}{N_j}\sum_{k=1}^{N_j}\overline p^{n_k}_t-\overline p_t\right|^\nu\,dt\right]=0.
\end{equation}
\end{lemma}
\begin{proof}
The convergence \eqref{app-X0} follows immediately from the convergence \eqref{app-xi0-2} and the definition of $X^{0,*}$. Taking $\chi=x$, $\zeta=\Lambda^1=-\Lambda^2=1/2\eta$, $\gamma=\Lambda^3=\varrho=\kappa$, $\Lambda^4=2\lambda$, $\Lambda^5=-\kappa\mathbb E\left[\left.\frac{1}{2\eta}\right|\mathcal F^0\right]$, $\overline f^n=0$ and $\overline g^n=\widetilde\kappa^0\xi^{0,n,*}$ in \eqref{FBSDE-appendix-n} the convergence \eqref{app-X} and \eqref{app-Y} follows from Theorem \ref{stability-appendix}, due to the uniform boundedness of $\overline g^n$ in $L^2$.

In  \eqref{FBSDE-appendix-n}, let $\chi=x^0$, $\Lambda^1=\Lambda^2=\zeta=1/2\eta$, $\gamma=\Lambda^3=\varrho=\kappa$, $\Lambda^4=2\lambda$, $\Lambda^5=\kappa\mathbb E\left[\left.\frac{1}{2\eta}\right|\mathcal F^0\right]$, $Q^n=-q^n$ and $(\overline f^n,\overline g^n)$ as in \eqref{input-f-bar-n} and \eqref{input-g-bar-n}. It follows from \eqref{app-X0}-\eqref{app-Y} that
\begin{equation}
    \lim_{N\rightarrow\infty}\mathbb E\left[\int_0^T\left|\frac{1}{\overline N}\sum_{j=1}^{\overline N}\frac{1}{N_j}\sum_{k=1}^{N_j}(\overline f^{n_k}_t,\overline g^{n_k})-(\overline f_t,\overline g_t)\right|^\nu\,dt\right]=0,
\end{equation}
 where $\overline f$ and $\overline g$ are defined as in \eqref{input-f-bar} and \eqref{input-g-bar}, respectively. By Corollary \ref{existence-X-Y-B} and the estimate \eqref{estimate-A}, we have $\overline f\in S^2_{\mathbb F}$ and $\overline g\in L^2_{\mathbb F}$. So \eqref{app-q} and \eqref{app-r} follow again from Theorem \ref{stability-appendix}. By \eqref{app-xi0-2}, \eqref{app-X}, \eqref{app-Y} and \eqref{app-q} we  also have \eqref{app-p-bar}.
\end{proof}

The preceding approximation lemma yields a representation on the candidate optimal strategy in terms of the candidate optimal state and adjoint processes akin to the maximum principle for the penalized problem.

\begin{theorem}\label{nec-charac-optim-stra}
The limit $\xi^{0,*}$ in \eqref{app-xi0-2} admits the following representation:
\begin{equation}\label{nec-charac-eq}
    \xi^{0,*}_t=\frac{p_t+\mathbb E[\widetilde\kappa^0_tq_t|\mathcal F^0_t]-\kappa^0_tX^{0,*}_t}{2\eta^0_t},\quad\textrm{a.e. a.s. on }[0,T]\times\Omega,
\end{equation}
where $p:=\overline AX^{0,*}+\overline p$.
Moreover, $\xi^{0,*}\in\mathcal A_{\mathbb F}(x^0)$ and $p$ satisfies the dynamic \eqref{leader-adjoint-p}.
\end{theorem}
\begin{proof}
The characterization \eqref{nec-charac-eq} follows immediately from Theorem \ref{p-max-principle} and Lemma \ref{app}. It remains to verify the admissibility of $\xi^{0,*}$. The fact that $\xi^{0,*}$ belongs to $L^2_{\mathbb F^0}$ is due to \eqref{xi0-L2}.
 By the uniform boundedness \eqref{uniform-boundedness-xi0n},
\[
    \lim_{n\rightarrow\infty}\mathbb E[(X^{0,n,*}_T)^2]=0.
\]
By \eqref{app-X0},
\begin{equation*}
   \lim_{\widetilde N\rightarrow\infty} \mathbb E\left[\left|\frac{1}{\widetilde N}\sum_{i=1}^{\widetilde N}\frac{1}{\overline N_i}\sum_{j=1}^{\overline N_i}\frac{1}{N_j}\sum_{k=1}^{N_j}X^{0,n_k,*}_T-X^{0,*}_T\right|^\nu\right]= 0.
\end{equation*}
Thus,
\begin{equation*}
    \begin{split}
    &~\mathbb E[|X^{0,*}_T|^\nu]\\
    \leq&~ 2\mathbb E\left[\left|\frac{1}{\widetilde N}\sum_{i=1}^{\widetilde N}\frac{1}{\overline N_i}\sum_{j=1}^{\overline N_i}\frac{1}{N_j}\sum_{k=1}^{N_j}X^{0,n_k,*}_T-X^{0,*}_T\right|^\nu\right]+2\frac{1}{\widetilde N}\sum_{i=1}^{\widetilde N}\frac{1}{\overline N_i}\sum_{j=1}^{\overline N_i}\frac{1}{N_j}\sum_{k=1}^{N_j}\mathbb E\left[|X^{0,n_k,*}_T|^\nu\right]\rightarrow 0,
    \end{split}
\end{equation*}
which implies $X^{0,*}_T=0$ a.s.. Finally, starting from $p:=\overline A X^{0,*}+\overline p$ by integration by parts and taking into account the characterization \eqref{nec-charac-eq}, we know $p$ satisfies \eqref{leader-adjoint-p}.
\end{proof}


\subsubsection{Sufficient maximum principle}

In this section, a sufficient maximum principle is established, from which we obtain the optimality of $\xi^{0,*}$ for the leader's MFC problem. The next theorem verifies that $\xi^{0,*}$ is indeed the unique optimal strategy for the leader.
\begin{theorem}[Sufficient maximum principle]\label{suff-max-prin}
Under the Assumption \ref{assumption-MFG}, $\xi^{0,*}$ given by Theorem \ref{nec-charac-optim-stra} is the unique optimal strategy to the leader's optimization.
\end{theorem}
\begin{proof}
We denote by $(X^{0,*},X^*,Y^*)$ the states corresponding to $\xi^{0,*}$ and by $(X^0,X,Y)$ the states corresponding to a generic strategy $\xi^0\in L^2_{\mathbb F^0}$. The verification is split into three steps.

\textbf{Step 1.}
By Corollary \ref{mu-convex-xi0}, $X$ and $Y$ are convex in $\xi^0$ in the sense that
 \[
    (X(\rho\xi^0+(1-\rho)\xi^{0'}),Y(\rho\xi^0+(1-\rho)\xi^{0'}))=\rho(X(\xi^0),Y(\xi^{0}))+(1-\rho)(X(\xi^{0'}),Y(\xi^{0'})).
 \]
 Thus, $J^0$ is strictly convex in $\xi^0$. As a result, there is at most one optimal strategy.

\textbf{Step 2.} Integration by part for $(X^0-X^{0,*})p$, $(X-X^*)r$ and $(Y-Y^*)q$ on $[0,\widetilde T]$ for $0\leq \widetilde T<T$ yields,
\begin{align*}\label{int-by-part}
        &~\mathbb E\left[(X^0_{\widetilde T}-X^{0,*}_{\widetilde T})p_{\widetilde T}\right]+\mathbb E\left[(X_{\widetilde T}-X^*_{\widetilde T})r_{\widetilde T}\right]+\mathbb E\left[(Y_{\widetilde T}-Y^*_{\widetilde T})q_{\widetilde T}\right]\\
         =&~-\mathbb E\left[\int_0^{\widetilde T}(X^0_t-X^{0,*}_t)\left(\overline\kappa^0_t\mathbb E\left[\xi^*_t|\mathcal F^0_t\right]+\kappa^0_t\xi^{0,*}_t+2\lambda^0_tX^{0,*}_t\right)\,dt\right]\\
         &~-\mathbb E\left[\int_0^{T-\epsilon}\mathbb E[\kappa_t(X_t-X^*_t)|\mathcal F^0_t]\left(-\overline\kappa^0_t\mathbb E\left[\left.\frac{1}{2\eta_t}\right|\mathcal F^0_t\right]X^{0,*}_t-2\overline\lambda_t\mathbb E\left[\left.\frac{1}{2\eta_t}\right|\mathcal F^0_t\right]\mathbb E[\xi^*_t|\mathcal F^0_t]\right)\,dt\right]\\
        &~-\mathbb E\left[\int_0^{\widetilde T}\mathbb E\left[\left.\frac{Y_t-Y^*_t}{2\eta_t}\right|\mathcal F^0_t\right]\left(\overline\kappa^0_tX^{0,*}_t+2\overline\lambda_t\mathbb E\left[\xi^*_t|\mathcal F^0_t\right]\right)\,dt\right]\\
        &~-\mathbb E\left[\int_0^{\widetilde T}(p_t+\mathbb E[\widetilde\kappa^0_tq_t|\mathcal F^0_t])(\xi^0_t-\xi^{0,*}_t)\,dt\right],
\end{align*}
where we recall $\xi^*$ is defined in \eqref{optimal-strategy-follower-2}.

\textbf{Step 3.} In order to prove the optimality of the strategy \eqref{nec-charac-eq} we define, for any $\widetilde T<T$ the cost functional
\begin{equation*}
    \begin{split}
    &~\widetilde J^0(\xi^0):=\mathbb E\left[\int_0^{\widetilde T}\overline\kappa^0_t\left(\mathbb E\left[\left.\frac{Y_t}{2\eta_t}\right|\mathcal F^0_t\right]-\mathbb E\left[\left.\frac{1}{2\eta_t}\right|\mathcal F^0_t\right]\mathbb E[\kappa_tX_t|\mathcal F^0_t]\right)X^0_t+\kappa_t^0\xi^0_tX^0_t+\eta^0_t(\xi^0_t)^2\right.\\
    &~\left.+\lambda^0_t(X^0_t)^2+\overline\lambda_t\left|\mathbb E\left[\left.\frac{Y_t}{2\eta_t}\right|\mathcal F^0_t\right]-\mathbb E\left[\left.\frac{1}{2\eta_t}\right|\mathcal F^0_t\right]\mathbb E[\kappa_tX_t|\mathcal F^0_t]\right|^2\,dt\right].
    \end{split}
\end{equation*}
By direct calculation we have
\begin{equation}\label{P-difference-cost}
    \begin{split}
    &~\widetilde J^0(\xi^0)-\widetilde J^0(\xi^{0,*})\\
    \geq&~ \mathbb E\left[\int_0^{T-\epsilon}(X^0_t-X^{0,*}_t)\left(\overline\kappa^0_t\mathbb E[\xi^*_t|\mathcal F^0_t]+\kappa^0_t\xi^0_t+2\lambda^0_t\xi^{0,*}_t\right)\,dt\right]\\
    &~+\mathbb E\left[\int_0^{T-\epsilon}\mathbb E\left[\left.\frac{Y_t-Y^*_t}{2\eta_t}\right|\mathcal F^0_t\right]\left(\overline\kappa^0_tX^{0,*}_t+2\overline\lambda_t\mathbb E[\xi^*_t|\mathcal F^0_t]\right)\,dt\right]\\
    &~+\mathbb E\left[\int_0^{T-\epsilon}\mathbb E[\kappa_t(X_t-X^*_t)|\mathcal F^0_t]\left(-\overline\kappa^0_t\mathbb E\left[\left.\frac{1}{2\eta_t}\right|\mathcal F^0_t\right]X^{0,*}_t-2\overline\lambda_t\mathbb E\left[\left.\frac{1}{2\eta_t}\right|\mathcal F^0_t\right]\mathbb E[\xi^*_t|\mathcal F^0_t]\right)\,dt\right]\\
    &~+\mathbb E\left[\int_0^{T-\epsilon}(\xi^0_t-\xi^{0,*}_t)\left(\kappa^0_tX^{0,*}_t+2\eta^0_t\xi^{0,*}_t\right)\,dt\right]
    \end{split}
\end{equation}
Plugging the result in Step 2 into \eqref{P-difference-cost} and taking into account the characterization \eqref{nec-charac-eq}, we have
\[
   \widetilde J^0(\xi^0)-\widetilde J^0(\xi^{0,*})+\mathbb E\left[(X^0_{\widetilde T}-X^{0,*}_{\widetilde T})p_{\widetilde T}\right]+\mathbb E\left[(X_{\widetilde T}-X^*_{\widetilde T})r_{\widetilde T}\right]+ \mathbb E\left[(Y_{\widetilde T}-Y^*_{\widetilde T})q_{\widetilde T}\right]\geq 0.
\]
The same estimate as in the proof of \cite[Theorem 2.9]{FGHP-2018} yields that
\begin{equation*}
    \begin{split}
    &~\lim_{\widetilde T\nearrow T}\mathbb E\left|(X^0_{\widetilde T}-X^{0,*}_{\widetilde T})p_{\widetilde T}\right|=0.
    \end{split}
\end{equation*}
Moreover, Corollary \ref{existence-X-Y-B} and Corollary \ref{existence-q-r-D} imply that
\begin{equation*}
    \begin{split}
        &~\mathbb E\left[(X_{\widetilde T}-X^*_{\widetilde T})r_{\widetilde T}\right]+ \mathbb E\left[(Y_{\widetilde T}-Y^*_{\widetilde T})q_{\widetilde T}\right]\\
        =&~\mathbb E\left[(X_{\widetilde T}-X^*_{\widetilde T})(-A_{\widetilde T}q_{\widetilde T}+D_{\widetilde T})+\left(A_{\widetilde T}X_{\widetilde T}+B_{\widetilde T}-A_{\widetilde T}X^*_{\widetilde T}-B^*_{\widetilde T}\right)q_{\widetilde T}\right]\\
        =&~\mathbb E\left[(X_{\widetilde T}-X^*_{\widetilde T})D_{\widetilde T}+(B_{\widetilde T}-B^*_{\widetilde T})q_{\widetilde T}\right]\\
        \rightarrow&~0,\qquad\textrm{ as }\widetilde T\nearrow T.
    \end{split}
\end{equation*}
Thus, letting $\widetilde T\nearrow T$, dominated convergence yields
\[
    J^0(\xi^0)-J^0(\xi^{0,*})\geq 0.
\]
\end{proof}

As a corollary, we obtain that a convex combination of the value functions for the penalized optimization problems converges to the value function of the constrained problem.

\begin{corollary}\label{app-cost}
There exists a convex combination of the value functions converging to $J^{0}(\xi^{0,*})$, i.e.,
\[
    \lim_{\widetilde N\rightarrow\infty}\frac{1}{\widetilde N}\sum_{i=1}^{\widetilde N}\frac{1}{\overline N_i}\sum_{j=1}^{\overline N_i}\frac{1}{N_j}\sum_{k=1}^{N_j}J^{0,n_k}(\xi^{0,n_k,*})=J^0(\xi^{0,*}).
\]
\end{corollary}
\begin{proof}
Recall that $X^{0,n_k,*}$ and $\xi^{n_k,*}$ are the optimal state of the leader and the optimal strategy of the follower corresponding to $\xi^{0,n_k,*}$, respectively. Due to the additional penalty term in the definition of $J^{0,n_k}$ and because $\xi^{0,*}$ is an admissible strategy for the penalized problem,\footnote{Notice that $J^ 0(\xi^{0,n_k,*})$ is well-defined even though $\xi^{0,n_k,*}$ may not not admissible for the constrained optimization problem.}
\begin{equation*}
    \begin{split}
        J^0(\xi^{0,n_k,*})
        \leq J^{0,n_k}(\xi^{0,n_k,*})=\inf_{\xi\in L^2_{\mathbb F^0}([0,T]\times\Omega;\mathbb R)}J^{0,n_k}(\xi)
        \leq J^0(\xi^{0,*})
    \end{split}
\end{equation*}

Denote by $K(\widetilde N)$ the cost functional with $(\xi^0,X^0,\xi)$ in $J^0$ replaced by $$\left(\frac{1}{\widetilde N}\sum_{i=1}^{\widetilde N}\frac{1}{\overline N_i}\sum_{j=1}^{\overline N_i}\frac{1}{N_j}\sum_{k=1}^{N_j}\xi^{0,n_k,*},\frac{1}{\widetilde N}\sum_{i=1}^{\widetilde N}\frac{1}{\overline N_i}\sum_{j=1}^{\overline N_i}\frac{1}{N_j}\sum_{k=1}^{N_j}X^{0,n_k,*},\frac{1}{\widetilde N}\sum_{i=1}^{\widetilde N}\frac{1}{\overline N_i}\sum_{j=1}^{\overline N_i}\frac{1}{N_j}\sum_{k=1}^{N_j}\xi^{n_k,*}\right).$$
By the convexity, we have
\begin{equation*}
    \begin{split}
    &~K(\widetilde N)
   \leq\frac{1}{\widetilde N}\sum_{i=1}^{\widetilde N}\frac{1}{\overline N_i}\sum_{j=1}^{\overline N_i}\frac{1}{N_j}\sum_{k=1}^{N_j}J^0(\xi^{0,n_k,*})\leq J^0(\xi^{0,*}).
    \end{split}
\end{equation*}
By Lemma \ref{app}, \eqref{app-xi0-2} and Fatou's lemma,
\[
    J^0(\xi^{0,*})\leq\liminf_{\widetilde N\rightarrow\infty}K(\widetilde N)\leq \liminf_{\widetilde N\rightarrow\infty}\frac{1}{\widetilde N}\sum_{i=1}^{\widetilde N}\frac{1}{\overline N_i}\sum_{j=1}^{\overline N_i}\frac{1}{N_j}\sum_{k=1}^{N_j}J^0(\xi^{0,n_k,*})\leq J^0(\xi^{0,*}).
\]
\end{proof}


\section{Conclusion}

We established existence and uniqueness of solutions results for linear McKean Vlasov FBSDEs with a terminal state constraint on the forward process. The general results were used to solve novel MFC problems and mean-field leader-follower games of optimal portfolio liquidation. For the leader-follower game it could be viewed as a MFC problem where the state dynamics follows a controlled FBSDE. For such problems we proved a novel stochastic maximum principle. The proof was based on a approximation method. We proved that both the sequence of optimal solutions and the sequence of state and adjoint equations associated with a family of penalized problems Cesaro converge to a unique limit that yields the optimal solution, respectively, the adjoint equations to the original state-constrained problem. To the best of our knowledge, no numerical methods for simulating the solution to conditional McKean-Vlasov FBSDEs are yet available. It would be desirable to develop such methods in order to study the interplay between the leader's and the follower's equilibrium strategies in greater detail.

\bibliography{bib_Guanxing_thesis}

\end{document}